\documentclass[preprint,3p]{elsarticle1}

\usepackage{graphicx}
\usepackage{amssymb}
\usepackage{comment}

\usepackage{lineno}

\usepackage{amsmath}
\usepackage{amsthm}
\usepackage{epsfig}

\usepackage{mathrsfs}

\usepackage{psfrag}
\usepackage{color}

\journal{Linear Algebra and its Applications}

\newtheorem{theorem}{Theorem}
\newtheorem{lemma}[theorem]{Lemma}
\newtheorem{example}[theorem]{Example}
\newtheorem{definition}[theorem]{Definition}

\newtheorem{corollary}[theorem]{Corollary}
\newtheorem{remark}{Remark}

\newtheoremstyle{algstyle}%
  {10mm}       % measure of space to leave above the theorem. E.g.: 3pt
  {10mm}       % measure of space to leave below the theorem. E.g.: 3pt
  {\tt}   % name of font to use in the body of the theorem
  {0pt}        % measure of space to indent
  {\bfseries}  % name of head font
  {\newline}   % punctuation between head and body
  {10mm}       % space after theorem head
  {\thmname{#1}\thmnumber{ #2}\thmnote{ (#3)}}          
\theoremstyle{algstyle}
\newtheorem{algorithm}{Algorithm}

\newtheoremstyle{algdashstyle}%
  {10mm}       % measure of space to leave above the theorem. E.g.: 3pt
  {10mm}       % measure of space to leave below the theorem. E.g.: 3pt
  {\tt}   % name of font to use in the body of the theorem
  {0pt}        % measure of space to indent
  {\bfseries}  % name of head font
  {\newline}   % punctuation between head and body
  {10mm}       % space after theorem head
  {\thmname{#1}\thmnumber{ #2}$'$\thmnote{ (#3)}}          % Manually specify head
\theoremstyle{algdashstyle}

\newcommand{\nw}[1]{%
\textbf{#1}%
}

\newcommand{\mnw}[1]{%
\boldsymbol{#1}%
}

\newcommand{\ppmatrix}[1]{%
\begin{pmatrix} #1 \end{pmatrix}%
}

\newcommand{\lrar}{\leftrightarrow}
\newcommand{\equivd}{:\equiv}

\newcommand{\equaln}{\hspace{0.1cm} = \hspace{0.1cm}}

\newcommand{\V}{\mbox{$\cal V$}} 

\newcommand{\F}{\mbox{$\cal F$}} 
    \newcommand{\0}{{\mathbf 0}}        
\newcommand{\Vsp}{{\cal V}_{SP}}           			%new  \cal V_S 
\newcommand{\Gsp}{{\cal G}_{SP}}           			%new  \cal V_S 
\newcommand{\Ksp}{{\cal K}_{SP}}           			%new  \cal V_S 
\newcommand{\Asp}{{\cal A}_{SP}}           			%new  \cal V_S 
\newcommand{\A}{{\cal A}}           			%new  \cal V_S 
\newcommand{\Ks}{{\cal K}_{S}}           			%new  \cal V_S 
\newcommand{\Kpq}{{\cal K}_{PQ}}           			%new  \cal V_S 
\newcommand{\Apq}{{\cal A}_{PQ}}           			%new  \cal V_S 
\newcommand{\Q}{{\mathbb{Q}}}

\newcommand{\Vpq}{{\cal V}_{PQ}}            			%new  \cal V_PQ 
\newcommand{\T}[0]{{\cal T}}                    	%         \cal D
\newcommand{\E}{\mbox{$\cal E$}} 
 
\newcommand{\G}[0]{{\cal G}}                       %        \cal G
  
\newcommand{\K}[0]{{\cal K}}                       %        \cal K
\newcommand{\N}[0]{{\cal N}}    							%new     \cal N
\newcommand{\B}{\mbox{${\cal B}$}}  				%new      \cal B
\newcommand{\g}{\mbox{${\bf g}$}}        				%new 		 \bf g
\newcommand{\I}{\mbox{${\bf I}$}}      				%new  \bf I 
\newcommand{\x}{\mbox{${\bf x}$}}             		%new  \bf x
\begin{document}

\begin{frontmatter}



%\title{Analysis of Linear Dynamical Systems without using State Space Representation}
\title{
%On the port behaviour of a linear electrical multiport}
On Thevenin-Norton and Maximum power transfer theorems}
%Implicit Linear Algebra and  Circuit Theory I: Thevenin and Max power


\author[hn]{H. Narayanan\corref{cor1}}
\ead{hn@ee.iitb.ac.in}
\cortext[cor1]{Corresponding author}
\author[hari]{Hariharan Narayanan}
\ead{hariharan.narayanan@tifr.res.in}
\address[hn]{Department of Electrical Engineering, Indian Institute of Technology Bombay}
\address[hari]{School of Technology and Computer Science, Tata Institute of Fundamental Research}

\begin{abstract}
In this paper we state and prove complete versions of two basic
theorems of linear circuit theory.

The Thevenin-Norton theorem expresses the port behaviour of 
a linear multiport in terms of the zero source and zero port input conditions
 when the port behaviour has a hybrid input- output representation of the form
\begin{align}
\label{eqn:abs1}
\ppmatrix{i_{P_1}\\v_{P_2}}=\ppmatrix{g_{11}&h_{12}\\h_{21}&r_{22}}\ppmatrix{v_{P_1}\\i_{P_2}}+\ppmatrix{J_{P_1}\\E_{P_2}},
\end{align}
{\it where the partition of $P$ into $P_1,P_2$ is known}.
%The characteristic feature of the theorem is that the final hybrid representation is computed through repeated solution of simplified versions of the original  
%circuit, namely the circuits obtained by setting the sources to zero and
%setting exactly one of the variables in $(v_{P1},i_{P_2})$ to $1$ and the rest to zero, and then the circuit obtained by setting all the variables in $(v_{P1},i_{P_2})$ to zero but keeping all the sources in place.
%This method fails when the multiport does not have a hybrid representation.

In this paper we show how to handle multiports 
which have port behaviour equations of the general form
$Bv_P-Qi_P=s,$
 which cannot even be put into the hybrid form of 
Equation \ref{eqn:abs1}, indeed may have number of equations ranging from
$0$ to $2n,$ where $n$ is the number of ports.
%We however assume that the multiport is consistent and for a given 
%port condition (voltage and current) has a unique internal solution 
%(voltage and current).
We do this 
 through repeatedly solving  with different source inputs,
a larger network obtained by terminating the multiport by its adjoint through a gyrator. The method works for linear multiports which are consistent 
for arbitrary internal source values and further have the property that the 
port conditions uniquely determine internal conditions.

The maximum power transfer theorem states that if the multiport has a Thevenin 
impedance matrix $Z,$ then the maximum power transfer from the multiport takes place when we terminate the multiport by another  whose 
Thevenin
impedance matrix is the adjoint (conjugate-transpose) $Z^*$ of $Z,$
provided $Z+Z^*$ has only positive eigenvalues \cite{desoer1}.
The theorem does not handle the case where the multiport does not have a 
Thevenin or Norton equivalent.

In this paper we present the most general version of maximum power transfer theorem
possible. This version of the theorem states that `stationarity' (derivative 
zero condition) of  power transfer occurs
when the multiport is terminated by its adjoint, provided the resulting network has a solution. If this network does not have a solution
there is no port condition for which stationarity holds.
This theorem does not require that the multiport has a 
hybrid immittance matrix.

%tmultiport decomposition method. 
%This method computes the port behaviour of the many individual multiports 
%into which the network is decomposed, and combines this with the  information about the 
%way in which the ports of the different multiports are connected, 
%to compute the actual voltages and currents occurring at the ports.
%Using these, the interior voltages and currents of the multiports are 
%calculated, which yields the solution of the original network. 
%At present the usual way of computing port behaviour is through 
%some variation of the Thevenin-Norton theorem, which does not work 
%for very general linear networks.
%
%In this paper we give a simple and efficient procedure for computing the port behaviour of a linear multiport by terminating it with its adjoint through 
%a gyrator. This computation can be done by using freely available 
%circuit simulators. The only assumption made is that the original 
%network has a unique solution for arbitrary source values.
%As a consequence, the decomposed multiports are `rigid', i.e., 
%have nonvoid solution for arbitrary source values and have unique 
%interior solution for a given port condition.
% 
%We also show that the maximum power that a multiport can transfer 
%across its ports corresponds to the port condition that obtains 
%when it is terminated by its adjoint through a $1:1$  ideal transformer.
%This yields the most general form of the maximum power transfer theorem.

\end{abstract}


\end{frontmatter}

%\linenumbers

\section{Introduction}
\label{sec:intro}
{\bf Note:}
This paper is a brief and self contained account of the basic circuit theorems extracted from 
`Implicit Linear Algebra and Basic Circuit Theory' (ILABC)
(arXiv:2005.00838v1 [eess.SY]). ILABC attempts to relate implicit 
linear algebra (ILA) and circuit theory and the discussion of the theorems 
in that paper is meant to be an application of ILA.
The present paper does not emphasize ILA but only concentrates on 
Thevenin-Norton  and Maximum Power
 Transfer Theorems.

In this paper we give complete versions of Thevenin-Norton Theorem and Maximum Power
 Transfer Theorem.

Thevenin-Norton Theorem gives a method for computing the port behaviour (formal definition
in Section \ref{sec:matched})
 of a linear multiport which has a hybrid representation (\cite{thevenin,mayer,norton}, see \cite{desoerkuh}
for a standard treatment), and for which we know before hand, which  
ports should be treated as current or voltage input kind.
%
%This problem, as is well known, was first solved for a special case by Thevenin
% and extended by others (\cite{thevenin0,mayer,norton}, see \cite{desoerkuh}
%for a standard treatment).
The characteristic feature of this approach is the computation of  the port condition 
for a special port termination (open circuit or short circuit), and then 
the computation of  the resistance, conductance or hybrid matrix setting the sources 
inside to zero. It is clear how to use a standard circuit simulator for this 
purpose, if one knew before hand that a hybrid matrix exists for the network
and which are its current and voltage input ports.
However, in general, a hybrid matrix may not even exist for a multiport.
A simple termination such as open or short circuit may result in inconsistency,
so that a standard circuit simulator would give an error message.

We present in this paper a method which works for very general linear multiports,
 which we call `rigid' (which may have number of port equations ranging from 
$0$ to twice the number of ports). {\it These have nonvoid solution for arbitrary 
internal source values and further have a unique internal solution for a 
given voltage plus current port condition.} These are also the only kind of 
multiports which can be handled by standard circuit simulators,
which are built for solving linear circuits with unique solution.
Our method involves terminating the given multiport by its adjoint (see Subsection \ref{sec:adjoint}) 
through  a multiport gyrator with identity matrix as the gyrator resistance matrix. The gyrator has external current or voltage sources attached 
in such a
way that the original multiport does not see a direct termination by a source
(see Figure \ref{fig:temp4}) and the circuit is solved repeatedly by the simulator.
%For rigid multiports the resulting network 
%always has a unique solution. Next the internal sources are set to zero 
%and external current or voltage sources are attached to the gyrator in such a
%way that the original multiport does not see a direct termination by a source
%(see Figure \ref{fig:temp4}) and the circuit is solved repeatedly by the simulator.
%The resulting port condition vectors generate the vector space translate of the
%affine port behaviour.
It is to be noted that the adjoint of a multiport can be built 
by changing the device characteristic block by block, usually without serious computation.

 We next use these ideas to present the most general form of the maximum
 power transfer theorem. We show that the maximum power transfer,
 if it occurs at all, corresponds to the port condition that is obtained 
when we terminate the multiport by its adjoint through a $ 1:1$ ideal transformer.
When the multiport is rigid, if the simulator fails to solve for this termination because there is no solution, it means that 
power transfer can be unbounded. If it fails to solve because the solution is non unique, it means that there are an infinity of  port conditions corresponding to
stationarity of power transfer. 
%In this case, either no power can be drawn 
%from the multiport or power transfer can be unbounded.

To summarize, we claim there are two significant contributions in the paper:
\begin{enumerate}
\item If we terminate a rigid multiport by its adjoint, through a $1:1$ gyrator,
the resulting network {\it always has a unique solution}.
This holds even if the gyrator has sources attached (current source in parallel, voltage source in series) to its ports. 
Further, every possible port condition of the multiport can be realized
by a suitable distribution of sources to the gyrator ports.
This yields a technique for computing 
the port behaviour of a rigid multiport, using a conventional 
circuit simulator repeatedly.
\item To find the maximum power transferred by a rigid multiport,
it is enough if we terminate it by its adjoint and solve the resulting
circuit. There is no need to compute its port behaviour, which is usually
a more cumbersome process.
\end{enumerate}

{\it Note on the notation:} We have to deal with solution spaces of equations 
 of the kind $Bv_P-Qi_{P"}=s,$
 rather than with those of the kind $v_P-Zi_{P"}=E.$ 
Instead of working with the adjoint (conjugate transpose) of a matrix $Z,$
we have to work with the adjoint (see Subsection \ref{sec:adjoint}) of the solution space of $Bv_P-Qi_{P"}=0.$
In order for our technique to be effective, we have to show that if 
$Bv_P-Qi_{P"}=s,$
 defines the port behaviour of $\N_P$ then the adjoint of 
the solution space of $Bv_P-Qi_{P"}=0,$
 is the port behaviour of the multiport 
$\N_P^{adj}$ which is the adjoint of $\N_P$ (Corollary  \ref{cor:adjointmultiport}).
We therefore have to use  notation not usual in circuit theory, dealing with operations on vector spaces such as 
sum, intersection, contraction, restriction, matched composition etc.,
instead of operations on matrices.

{\bf Example.}
It is possible to build a rigid (as defined above) multiport $\N_P,$  
using controlled sources, norators and
 nullators,
 which has any given representation at the ports with nonvoid solution.
Suppose the rigid multiport $\N_P$ has the behaviour (see formal definition 
in Section \ref{sec:matched})
$Bv_P-Qi_{P"}=s.$
The number of equations may range from $0$ to $2n,$
 where $n=|P|.$
%\begin{align}
%\label{eqn:3}
%\ppmatrix{i_{P1}\\v_{P_2}}=\ppmatrix{g_{11}&h_{12}\\h_{21}& r_{22}}\ppmatrix{v_{P1}\\i_{P_2}}+\ppmatrix{J_{P1}\\E_{P_2}},
%\end{align}
%We will suppose that the partition of $P$ into $P_1,P_2,$ is unknown and that
% the submatrices $g_{11},r_{22}$ are singular. In this case, 
Neither the 
Thevenin nor the Norton representation, indeed even a hybrid
representation, need exist.
{\it Our method can be used to compute the port behaviour
of this multiport,
whereas the usual Thevenin-Norton theorem cannot handle a case 
of this level of generality.}

\begin{figure}
\begin{center}
 \includegraphics[width=5.5in]{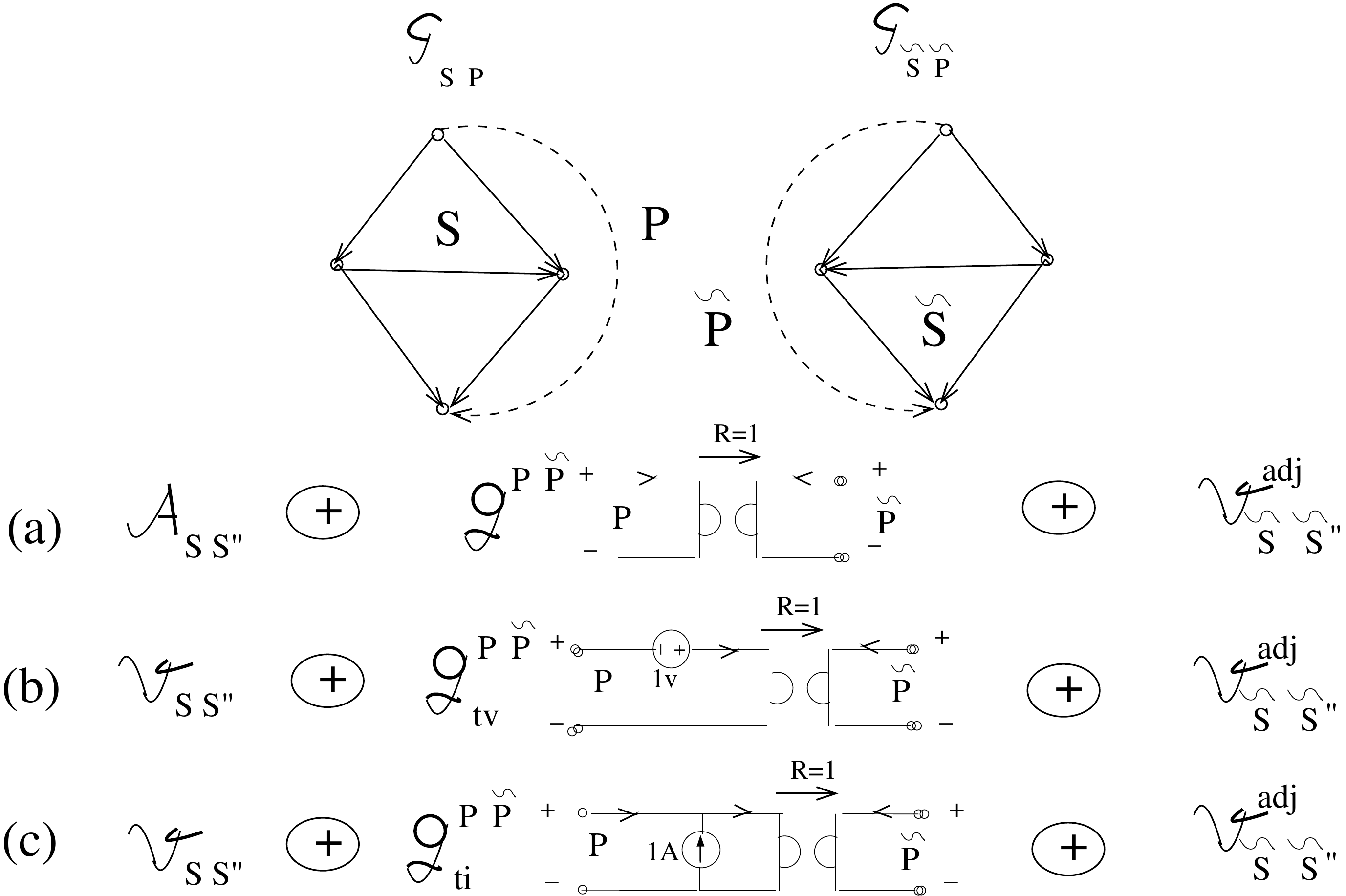}
 \caption{Computation of port behaviour of a multiport}

\label{fig:temp4}
%\caption{}
% \caption{Example $RLCEJ$ Network and its multiport decomposition
%into capacitive multiport and $RL$ multiport}
\end{center}
\end{figure}

Briefly, the steps in our method are as follows:
\\1. Build the adjoint multiport $\N^{adj}_{\tilde{P}}.$
This is done by retaining the same graph as $\N_P$ but replacing 
each device by its adjoint (see Example \ref{eg:egadjoint1}). The name of the set of internal edges 
$S$ is changed to $\tilde{S}$ and that of the port edges $P$
is changed to $\tilde{P}.$
\\2. Between every pair of corresponding ports $P_j, \tilde{P}_j,$
insert a source accompanied $1:1$ gyrator. (For simplicity, in Figure \ref{fig:temp4} we have taken $\N_P$ 
to be a $1$-port.)
\\3. Solve the circuit resulting when all the gyrator sources
are zero and internal port sources are active as in Figure \ref{fig:temp4} (a).
Let $(v^p_{P},i^p_{P"})$ be the corresponding port condition.
\\4. Set internal sources to zero but gyrator accompanied sources 
active one at a time as in Figure \ref{fig:temp4} (b),(c).
 The resulting set of voltage- current vectors at $P$ will span the  source 
free port behaviour $\hat{\V}_{P{P"}}.$
(We have to change the sign of the current port vector in the solution
to correspond to port behaviour current, which `enters' the multiport).
\\5. The port behaviour $\hat{\A}_{P{P"}}\equivd (v^p_{P},-i^p_{P"})+\hat{\V}_{P{P"}}.$
 
Further, we can also test this port behaviour for condition of maximum
power transfer by terminating $\N_P$ by its adjoint $\N^{adj}_{\tilde{P}}$ through a $1:1$ 
ideal transformer ($v_{P}=v_{\tilde{P}}; i_{P}=-i_{\tilde{P}}$) and solving the 
resulting circuit by a simulator.
Solving this  circuit is equivalent to solving
Equation \ref{eqn:optprob51} of Section \ref{subsec:maxpower} 
%in the revised manuscript 
rewritten as
Equation \ref{eqn:optprob511}
 below. ({\it Note that this equation is not being explicitly constructed
 and solved.})
\begin{align}
\label{eqn:optprob511}
\ppmatrix{B&-Q}\ppmatrix{-Q^*\\B^*}\lambda=s. 
\end{align}
If this equation has a unique solution, the corresponding port condition 
$${((\breve{v}^{stat}_{P})^T|(\breve{i}^{stat}_{P"})^T)=\lambda^T(-\overline{Q}|\overline{B}})$$
is the unique stationarity point. We have to verify whether it is maximum 
or minimum by perturbing the port condition around this point.
If this equation has no solution,
% if the vector $s\ne 0,$ and an infinity
%of solutions when $s=0.$ In the former case, 
the power that can be drawn from the multiport is unbounded.
If it has many solutions, there will be an infinity of port conditions 
corresponding to stationarity of power transfer.
In both these latter cases our method will only indicate non unique solution and halt.

Finally, note that when the Thevenin equivalent exists, the port equation
 $Bv_P-Qi_{P"}=s,$ reduces to $Iv_P-Zi_{P"}=E,$ and the solution of Equation \ref{eqn:optprob511} (taking $Z^*\equivd $ conjugate transpose of $Z$) yields \begin{align}
\label{eqn:optprob512}
\ppmatrix{I&-Z}\ppmatrix{-Z^*\\I}\lambda=E. 
\end{align}
This corresponds to 
$ -(Z+Z^*)\lambda =E, ((\breve{v}^{stat}_{P})^T|(\breve{i}^{stat}_{P"})^T)=\lambda^T(-\overline{Z}|I),$ i.e., to $ (Z+Z^*)(-\breve{i}^{stat}_{P"}) =E.$
This means that the stationarity of  power transfer occurs when we terminate the multiport by the adjoint of the Thevenin impedance.

\section{Preliminaries}
\label{sec:Preliminaries}
The preliminary results and the notation used are from \cite{HNarayanan1997}.
% (open 2nd edition available at \cite{HNarayanan2009}).
% and more specifically from
%\cite{HNPS2013}. 
\subsection{Vectors}
A \nw{vector} $\mnw{f}$ on a finite set $X$ over $\mathbb{F}$ is a function $f:X\rightarrow \mathbb{F}$ where $\mathbb{F}$ is either the real field $\Re$ or the complex field $\mathbb{C}.$  
%In this paper, we work only  with the rational field $\mathbb{Q}.$

%The \nw{length} of a vector $x$ is the Euclidean norm $||x||$ of $x.$

%

%The size of a set $X$ is denoted by $\mnw{|X|}.$
%and would often be written as $(f_X,f_Y)$ during operations dealing with such vectors. 
The {\it sets on which vectors are defined are  always  finite}. When a vector $x$ figures in an equation, we use the 
convention that $x$ denotes a column vector and $x^T$ denotes a row vector such as
in `$Ax=b,x^TA^T=b^T$'. Let $f_Y$ be a vector on $Y$ and let $X \subseteq Y$. The \textbf{restriction $f_Y|_X$} of $f_Y$ to $X$ is defined as follows:\\
%\begin{align*}
$f_Y|_X \equivd g_X, \textrm{ where } g_X(e) = f_Y(e), e\in X.$

%\end{align*}

When $f$ is on $X$ over $\mathbb{F}$, $\lambda \in \mathbb{F},$ then  the \nw{scalar multiplication} $\mnw{\lambda f}$ of $f$ is on $X$ and is defined by $(\lambda f)(e) \equivd \lambda [f(e)]$, $e\in X$. When $f$ is on $X$ and $g$ on $Y$ and both are over $\mathbb{F}$, we define $\mnw{f+g}$ on $X\cup Y$ by \\
%\begin{align*}
% (f+g)(e) &\equiv \left\{ \begin{matrix}
%                          f(e) + g(e),& e\in X \cap Y\\
%			    f(e),& e\in X \setminus Y\\
%			    g(e),& e\in Y \setminus X.
%                         \end{matrix}
%\right.
%\end{align*}
$(f+g)(e)\equivd f(e) + g(e),e\in X \cap Y,\ (f+g)(e)\equivd  f(e), e\in X \setminus Y,
\ (f+g)(e)\equivd g(e), e\in Y \setminus X.
$
(For ease in readability, we will henceforth use $X-Y$ in place of $X \setminus Y.$)

%The size of a set $X$ is denoted by $\mnw{|X|}.$
When $f,g$ are on $X$ over $\mathbb{C},$ the \textbf{inner product} $\langle f, g \rangle$ of $f$ and $g$ is defined by 
%\begin{align*}
$ \langle f,g \rangle \equivd \sum_{e\in X} f(e)\overline{g(e)},$
$ \overline{g(e)}$ being the complex conjugate of $g(e).$
If the field is $\mathbb{R},$ the inner product would reduce to the \nw{dot product} $ \langle f,g \rangle \equivd \sum_{e\in X} f(e){g(e)}.$

%\end{align*}
When $X, Y, $ are disjoint,  $f_X+g_Y$ is written as  $\mnw{(f_X, g_Y)},$
The disjoint
union of $X$ and $Y.$ is denoted by  $\mnw{X\uplus Y}.$
 A vector $f_{X\uplus  Y}$ on $X\uplus Y$ is  written as $\mnw{f_{XY}}.$

We say $f$, $g$ are \textbf{orthogonal} (orthogonal) iff $\langle f,g \rangle$ is zero.
%We say $f$, $g$ are \textbf{$\mathbb{Z}$-orthogonal}  iff $\langle f,g \rangle$ is an integer.

An \nw{arbitrary  collection} of vectors on $X$ 
%with $0_X$ as a member 
is denoted by $\mnw{\mathcal{K}_X}$. 
When $X$, $Y$ are disjoint we usually write $\mathcal{K}_{XY}$ in place of $\mathcal{K}_{X\uplus Y}$.
We write $\K_{XY}\equivd \K_X\oplus \K_Y$ iff
$\K_{XY}\equivd\{f_{XY}:f_{XY}=(f_X,g_Y), f_X\in \K_X, g_Y\in \K_Y\}.$
We refer to $\K_X\oplus \K_Y$ as the \nw{direct sum} of $\K_X, \K_Y.$ 

A collection $\K_X$ is a \nw{vector space} on $X$ iff it is closed under 
addition and scalar multiplication. 
The notation $\mnw{\V_X}$ always denotes
a vector space on $X.$
For any collection $\K_X,$  $\mnw{span(\K_X)}$ is the vector space of all
linear combinations of vectors in it.
We say $\A_X$ is an \nw{affine space} on $X,$ iff it can be expressed as
$x_X+\V_X \equivd \{y_X, y_X=x_X+z_X, z_X\in \V_X\},$ where $x_X$ is a vector and $\V_X,$ a vector space on $X.$
The latter is unique for $\A_X$ and is said to be its \nw{vector space translate}.

For a vector space  $\V_X,$ since we take $X$ to be finite,
any maximal independent subset of $\V_X$ has size less than or equal to $|X|$ and this 
size can be shown
to be unique. A maximal independent subset of a vector
space $\V_X$ is called its \nw{basis} and its  size 
is called the  {\bf dimension} or \nw{rank} of $\V_X$ and denoted by ${\mnw{dim}(\V_X)}$
 or by ${\mnw{r}(\V_X)}.$
For any collection of vectors $\K_X,$
the rank $\mnw{r}(\K_X)$
is defined to be $dim(span(\K_X)).$
The collection of all linear combinations of the rows of a matrix $A$ is a vector space 
that is denoted by $row(A).$

For any collection of vectors
$\mathcal{K}_X,$   the collection $\mnw{\mathcal{K}_X^{*}}$ is defined by
%\begin{align*}
$ {\mathcal{K}_X^{*}} \equivd \{ g_X: \langle f_X, g_X \rangle =0\},$
%f_X\in \mathcal{K}_X 
%\}.
%\end{align*}
It is clear that $\mathcal{K}_X^{*}$ is a vector space for 
any $\mathcal{K}_X.$ When $\mathcal{K}_X$ is a vector space $\V_X,$
 and the underlying set $X$ is finite, it can be shown that $({\mathcal{V}_X^{*}})^{*}= \mathcal{V}_X$ 
and  $\mathcal{V}_X,{\mathcal{V}_X^{*}}$ are said to be \nw{complementary orthogonal}. 
The symbol $0_X$ refers to the \nw{zero vector} on $X$ and $\mnw{0_X}$  refers to the \nw{zero vector space} on $X.$ The symbol $\mnw{\F_X}$  refers  to the collection of all vectors on $X$ over the field in question.
It is easily seen, when $X,Y$ are disjoint, and $\K_X, \K_Y$  
contain zero vectors, that $(\K_X\oplus \K_Y)^{*}=
\K_X^{*}\oplus\K_Y^{*}.$

%An \nw{arbitrary  collection} of vectors on $X$ with $0_X$ as a member would be denoted by $\mnw{\mathcal{K}_X}$. 
The  \nw{adjoint}  $K^*$ of a matrix $K$ is  defined by
$K^*\equivd \overline{K}^T.$\\
(Later we define the adjoint $\V^{adj}_{PP"}$ of a vector space 
which extends  this notion (see Subsection \ref{sec:adjoint}).)

A matrix of full row rank, whose rows generate a vector space $\V_X,$
is called a \nw{representative matrix} for $\V_X.$
A representative matrix which can be put in the form $(I\ |\ K)$ after column
permutation, is called a {standard representative matrix}.
It is clear that every vector space has a  standard representative matrix.
If $(I\ |\ K)$ is  a standard representative matrix of $\V_X,$
it is easy to see  that $(-K^*|I)$ is a standard representative matrix of $\V^{*}_X.$
Therefore we must have
\begin{theorem}
\label{thm:perpperp}
Let $\V_X$ be a vector space on $X.$ Then
$r(\V_X)+r(\V^{*}_X)=|X|$ and $((\V_X)^{*})^{*}=\V_X.$
\end{theorem}

%If the rows  of a matrix generate a number lattice $L_X,$
%by integral linear combination, then the matrix is called a \nw{generating matrix}
%for $L_X.$ If further the rows are linearly independent, the generating
%matrix is called a \nw{basis matrix} for $L_X.$
%When the index set is clear from the context, to improve readability, we sometimes
%represent a vector space as  $\{[K]\}$ or as $\{Kx=0\}.$ In the former case the vector
%space referred to is the row space of $K$ and in the latter  case, it is the solution space of $Kx=0.$ 

The collection
$\{ (f_{X},\lambda f_Y) : (f_{X},f_Y)\in \mathcal{K}_{XY} \}$
is denoted by
$ \mnw{\mathcal{K}_{X(\lambda Y)}}.
$
When $\lambda = -1$ we would write $ {\mathcal{K}_{X(\lambda Y)}}$  more simply as $\mnw{\mathcal{K}_{X(-Y)}}.$
Observe that $(\mathcal{K}_{X(-Y)})_{X(-Y)}=\mathcal{K}_{XY}.$
%$\Ipp $ is defined to be  the vector space $ \{ (f_P,f_P'):f_P\in \mathscr{F}_P\}.$ 
%\end{align*}
%\subsection{Building copies}
%\subsection{Copies}
We say sets $X$, $X"$ are \nw{copies of each other} iff they are disjoint and there is a bijection, usually clear from the context, mapping  $e\in X$ to $e"\in X"$.
When $X,X"$ are copies of each other, the vectors $f_X$ and $f_{X"}$ are said to be copies of each other with  $f_{X"}(e") \equivd  f_X(e), e \in X.$ 
The copy $\K_{X"}$ of $\K_X$ is defined by
 $\K_{X"}\equivd\{f_{X"}:f_X\in \K_X\}.$
When $X$ and $X"$ are copies of each other, the notation for interchanging the positions of variables with index sets $X$ and $X"$ in a collection $\mathcal{K}_{XX"Y}$ is given by $\mnw{(\mathcal{K}_{XX"Y})_{X"XY}}$, that is\\
$(\mathcal{K}_{XX"Y})_{X"XY}
 \equivd \{(g_X,f_{X"},h_Y)\ :\ (f_X,g_{X"},h_Y) \in \mathcal{K}_{XX"Y},\ g_X\textrm{ being copy of }g_{X"},\ f_{X"}\textrm{ being copy of }f_X  \}.$
An affine space $\mathcal{A}_{XX"}$ is said to be {\bf proper}
iff the rank of its vector space translate is $|X|=|X"|.$

%\section{Basic operations}
%\label{sec:basic}
%The basic operations we use in this paper are as follows:

\subsection{Sum and Intersection}
Note: This extends the conventional definition of sum and intersection
of vector spaces on the {\it same} set. \\
Let $\mathcal{K}_{SP}$, $\mathcal{K}_{PQ}$ be collections of vectors on sets $S\uplus P,$ $P\uplus Q,$ respectively, where $S,P,Q,$ are pairwise disjoint. The \nw{sum} $\mnw{\mathcal{K}_{SP}+\mathcal{K}_{PQ}}$ of $\mathcal{K}_{SP}$, $\mathcal{K}_{PQ}$ is defined over $S\uplus P\uplus Q,$ as follows:\\
 $\mathcal{K}_{SP} + \mathcal{K}_{PQ} \equivd  \{  (f_S,f_P,0_{Q}) + (0_{S},g_P,g_Q), \textrm{ where } (f_S,f_P)\in \mathcal{K}_{SP}, (g_P,g_Q)\in \mathcal{K}_{PQ} \}.$\\
%When $S$, $Q,$ are disjoint, $\mathcal{K}_S + \mathcal{K}_{Q}$ is usually written in this paper as $\mnw{\mathcal{K}_S \oplus \mathcal{K}_Q}$ and is called the \nw{direct sum}.
Thus,
$\mathcal{K}_{SP} + \mathcal{K}_{PQ} \equivd (\mathcal{K}_{SP} \oplus \0_{Q}) + (\0_{S} \oplus \mathcal{K}_{PQ}).$\\
%\subsubsection{Intersection}
The \nw{intersection} $\mnw{\mathcal{K}_{SP} \cap \mathcal{K}_{PQ}}$ of $\mathcal{K}_{SP}$, $\mathcal{K}_{PQ}$ is defined over $S\uplus P\uplus Q,$ where $S,P,Q,$ are pairwise disjoint, as follows:
$\mathcal{K}_{SP} \cap \mathcal{K}_{PQ} \equivd \{ f_{SPQ} : f_{S P Q} = (f_S,h_P,g_{Q}),$
%f_{(S\cup Y)} = (y_{(S\setminus Y)},f_Y), 
 $\textrm{ where } (f_S,h_P)\in\mathcal{K}_{SP}, (h_P,g_Q)\in\mathcal{K}_{PQ}.%  x_{Q}\in \F_{Q},y_{S} \in  \F_{ S}
\}.$\\
Thus,
$\mathcal{K}_{SP} \cap \mathcal{K}_{PQ}\equivd (\mathcal{K}_{SP} \oplus  \F_{Q}) \cap (\F_{S} \oplus \mathcal{K}_{PQ}).$\\

It is immediate from the definition of the operations that sum and intersection of
vector spaces remain vector spaces.

%The following identity is useful.
%
%====================================================
%
%change in other files also
%
%=========================================
%\begin{theorem}
%\label{thm:sumintersection}
%Let $\V^1_A, \V^2_B, \V_S,\V'_S $ be vector spaces. Then\\
%\begin{enumerate}
%\item $r(\V_S)+r(\V'_S)=r(\V_S+\V'_S)+r(\V_S\cap \V'_S);$
%\item $(\V^1_A+\V^2_B)^{\*}=(\V^1_A)^{\*}\cap (\V^2_B)^{\*};$
%\item $(\V^1_A\cap \V^2_B)^{\perp}=(\V^1_A)^{\*}+ (\V^2_B)^{\*}.$
%\end{enumerate}
%\end{theorem}

\subsection{Restriction and contraction}
%We remind the reader that
%$\mathcal{K}_X$ denotes an arbitrary collection of vectors on $X$ with a zero vector as a member.

The \nw{restriction}  of $\mnw{\mathcal{K}_{SP}}$ to $S$ is defined by
$\mnw{\mathcal{K}_{SP}\circ S}\equivd \{f_S:(f_S,f_P)\in \mathcal{K}_{SP}\}.$
The \nw{contraction}  of $\mnw{\mathcal{K}_{SP}}$ to $S$ is defined by
$\mnw{\mathcal{K}_{SP}\times S}\equivd \{f_S:(f_S,0_P)\in \mathcal{K}_{SP}\}.$
Unless otherwise stated, the sets on which we perform the contraction operation would 
have the zero vector as a member so that the resulting set would be nonvoid.
We denote by  $\mnw{\mathcal{K}_{SPZ}\circ SP}$, $\mnw{\mathcal{K}_{SPZ} \times SP}$, respectively
when $S,P,Z,$ are pairwise disjoint,  the collections of vectors   $\mnw{\mathcal{K}_{SPZ}\circ (S\uplus P)}$, $\mnw{\mathcal{K}_{SPZ} \times (S \uplus P)}.$

It is clear that restriction and contraction of vector spaces are also
vector spaces.

%For a graph $\G\equivd (V,E)$ with vertex set $V$ and edge set $E,$  the \nw{restriction}
%of $\G$ to $T\subseteq E, $ is denoted by $\G\circ T$ and is defined to be the 
%graph with edge set $T,$ obtained from $\G$ by deleting (open circuiting) the edges $E-T$ leaving the endpoints in place and then deleting the vertices which have no edges incident on them. The \nw{contraction} of $\G$ to $T\subseteq E, $ is denoted by $\G\times T$ and is defined to be the
%graph with edge set $T,$ obtained from $\G$ by deleting the edges $E-T$ but fusing their endpoints (short circuiting).

The following is a useful result on ranks of restriction and contraction 
of vector spaces and on relating restriction and  contraction through orthogonality \cite{tutte,HNarayanan1997}.
\begin{theorem}
\label{thm:dotcrossidentity}
%\begin{enumerate}
1. $r(\Vsp)=r(\Vsp\circ S)+r(\Vsp\times P).$ 2. 
$\V_{SP}^{*}\circ P= (\V_{SP}\times P)^{*};$
 $\V_{SP}^{*}\times S= (\V_{SP}\circ S)^{*}.$
\end{theorem}
\subsection{KVL and KCL}
{\it By a `graph' we  mean a `directed graph'.}
\nw{Kirchhoff's Voltage Law (KVL)} for a  graph states that the sum of the signed voltages of
edges 
around an oriented loop is zero - the sign of the voltage of an edge 
being positive if the edge orientation agrees with the orientation of the loop
and negative if it opposes. 
\nw{Kirchhoff's Current Law (KCL)} for a  graph states that the sum of the signed currents leaving 
a node is zero, the sign of the current of an edge being positive if 
 its positive endpoint is the node in question.

Let $\G$ be a graph with edge set $S.$ We refer to the space of vectors $v_{S},$ which satisfy Kirchhoff's Voltage Law (KVL) of the graph $\mathcal{G},$
by $\mnw{\V^v(\mathcal{G})}$ and to the space of vectors $i_{S"},$ which satisfy Kirchhoff's Current Law (KCL) of the graph $\mathcal{G},$
by $\mnw{\V^i(\mathcal{G})}.$ (We need to deal with vectors of the kind $(v,i).$ To be consistent with our notation for vectors as functions, this is treated  as a vector $(v_S,i_{S"}),$  where $S"$ is a disjoint copy of $S.$ Therefore $\mnw{\V^v(\mathcal{G})}$ is taken to be on set $S$ and $\mnw{\V^i(\mathcal{G})}$ is taken to be on a disjoint  copy  $S".$)
%{\it These vector spaces will, unless otherwise stated, be taken as  
%over $\Re.$}

The following is a fundamental result  on vector spaces associated with graphs.
\begin{theorem}
\label{thm:tellegen}
{\bf Tellegen's Theorem} \cite{tellegen}  Let $\G$ be a graph with edge set S, Then $\V^i(\mathcal{G})= (\V^v(\mathcal{G})^{*})_{S"}.$
\end{theorem}
%\begin{lemma}
%\label{lem:minorgraphvectorspace}
%\cite{tutte, HNarayanan1997} Let $\G$ be a graph on edge set $S.$ 
%Let $W\subseteq T\subseteq S.$
%\begin{enumerate}
%\item $ \V^v(\mathcal{G}\circ T)= (\V^v(\mathcal{G}))\circ T, \ \ \  \V^v(\mathcal{G}\times T)= (\V^v(\mathcal{G}))\times T,\ \ \V^v(\mathcal{G}\circ T\times W)= (\V^v(\mathcal{G}))\circ T \times W;$
%\item $ \V^i(\mathcal{G}\circ T)= (\V^i(\mathcal{G}))\times T, \ \ \  \V^i(\mathcal{G}\times T)= (\V^i(\mathcal{G}))\circ T,
%\V^i(\mathcal{G}\times T\circ W)= (\V^i(\mathcal{G}))\circ T \times W.$
%\end{enumerate}
%\end{lemma}
%(Proof also available at \cite{HNarayanan1997}).\\
%Note that $\times,\circ$ are graph operations on the left side of the equations
%and vector space operations on the right side.
\subsection{Networks and multiports}
\label{subsec:networks}
An  {\bf electrical network $\mathcal{N},$} or a `network' in short, is a pair $(\mathcal{G},\mathcal{K}),$ where $\mathcal{G}\equivd (V(\G),E(\G))$ is a graph
and $\mathcal{K}, $ called the \nw{device characteristic} of the network, is a collection of pairs of vectors  $(v_{S},i_{S"}),S\equivd E(\G),$ where $S,S"$ 
are disjoint copies of $S,\ $ $v_{S},i_{S"}$ are real or complex vectors on the edge set of the graph. 
%We call $\mathcal{D},$  the `device characteristic' of the network.
In this paper,  we deal only with affine device characteristics 
and with complex vectors, unless otherwise stated. When the device characteristic $\K_{SS"}$ is affine, we denote it by $\A_{SS"}$ and 
 say the network is \nw{linear}. If $\V_{SS"}$ is the vector space 
translate of $\A_{SS"},$ we say that $\A_{SS"}$ is the \nw{source accompanied}
 form of $\V_{SS"}.$ An affine space $\A_{SS"}$ is said to be \nw{proper} 
iff its vector space
translate $\V_{SS"}$ has dimension $|S|=|S"|.$ 

Let $S$ denote the set of edges of the graph of the network and let
$\{S_1, \cdots , S_k\} $ be a partition of $S,$ each block $S_j$ being an 
\nw{individual device}. Let $S,S"$ be disjoint copies of  $S,$
with $e,e"$ corresponding to edge $e.$ The device characteristic would usually have the form $\bigoplus \A_{S_jS_j"},$  
defined by $(B_jv_{S_j}+Q_ji_{S_j"})=s_j,$
with rows of $(B_j|Q_j)$ being linearly independent. The vector space translate
would have the form $\bigoplus \V_{S_jS_j"},$
$\V_{S_jS_j"}$ being the translate of $\A_{S_jS_j"}.$
Further, usually the blocks $S_j$ would have size one or two,
even
when the network has millions of edges.
Therefore, it would be easy to build $(\bigoplus \V_{S_jS_j"})^{*}=\bigoplus \V_{S_jS_j"}^{*}.$
\\
We say $S_j$ is  a set of \nw{norators}
iff $\A_{S_jS_j"}\equivd \F_{S_jS_j"},$ i.e., there are no constraints on
$v_{S_j},i_{S_j"}.$\\
We say $S_j$ is  a set of \nw{nullators}
iff $\A_{S_jS_j"}\equivd \0_{S_jS_j"},$ i.e.,
$v_{S_j}=i_{S_j"}=0.$\\
We say $\V_{S_j\hat{S}_jS_j"\hat{S}_j"}$ is a \nw{gyrator}
iff
$v_{S_j}=-R i_{\hat{S}_j"}; v_{\hat{S}_j}=R i_{S_j"},$ where $R$  is a positive diagonal matrix.
We denote by $\g ^{S_j\hat{S}_j},$ the gyrator where $R$ is the identity matrix.
%Further, usually the blocks $E_j$ would have size not more than ten, even 
%when the network has millions of nodes.

%We permit networks to be defined in terms of vector spaces rather than in terms
%of graphs when a network is derived from another. In this case a vector space
%$\Vsp$ would be in place of $\V^v(\mathcal{G})$  and $(\Vsp)^{\*}$ in place of $\V^i(\mathcal{G}).$
A {\bf solution} of $\mathcal{N}\equivd (\G,\K)$ on graph $\G\equivd (V(\G),E(\G))$ is a pair
 $(v_{S},i_{S"}),S\equivd E(\G)$ satisfying\\
$v_{S}\in \V^v(\mathcal{G}), i_{S"} \in \V^i(\mathcal{G})$
  (KVL,KCL)  and $(v_{S},i_{S"})\in \mathcal{K}.$
The KVL,KCL conditions are also called \nw{topological} constraints.
Let  $S,S"$ be disjoint copies of $S,$
let $\V_{S}\equivd \V^v(\mathcal{G}),$ so that, by Theorem \ref{thm:tellegen}, $(\V^{*}_{S})_{S"}= \V^i(\mathcal{G}),
$ and let $ \K_{SS"}$ be the device characteristic of $\N.$
The set of solutions of  $\mathcal{N}$ may be written as, 
$$\V_{S}\cap (\V^{*}_{S})_{S"}\cap \K_{SS"}=[\V_{S}\oplus (\V^{*}_{S})_{S"}]\cap \K_{SS"}.$$
This has the form `[Solution set of topological constraints] $\cap$ [Device characteristics]'.
%We would often abuse the notation by identifying $\N$ with $\V_{S'}\cap (\V^{\perp}_{S'})_{S"}\cap \K_{S'S"}.$
%The \nw{power absorbed} at an edge $e\in E$ corresponding to a solution 
%$(v_{E},i_{E"})$ is given by $\langle v_{E},i_{E"}\rangle +  \langle i_{E"},v_{E}\rangle$ (we omit the scale factor $\frac{1}{2}$ for better readability of the expressions involved).
%

A \nw{multiport} $\mathcal{N}_P\equivd (\Gsp,\K_{SS"})$ is a network with some subset $P$ of its
edges which are norators, specified as `ports'.
Let $\N_P$ be on graph $\G_{SP}$ with device characteristic $\K.$
Let $\V_{SP}\equivd \V^v(\G_{SP}),$ so that $ (\V^{*}_{SP})_{S"P"}= (\V^i(\G_{SP}))_{S"P"},$ and let $\K_{SS"}, $ be the  device characteristic 
on the edge set $S.$ The device characteristic of a  multiport $\mathcal{N}_P$
would be $\K\equivd \K_{SS"} \oplus \F_{PP"}.$
For simplicity we would refer to $\K_{SS"}$ as the device characteristic 
of $\N_P.$
The multiport is said to be \nw{linear} iff its device characteristic is affine.
%\\
%Let 
%$\{S_1, \cdots , S_k\} $ be a partition of $S,$ each block $S_j$ being an
%\nw{individual device}. Let $S,S"$ be disjoint copies of  $S,$
%with $e,e"$ corresponding to edge $e.$ The device characteristic of a linear multiport would usually have the form $\bigoplus \A_{S_jS_j"},$
%defined by $(B_j'v_{S_j}+Q_j"i_{S_j"})=s_j,$
%with rows of $(B_j'|Q_j")$ being linearly independent.
% The vector space translate
%would have the form $\bigoplus \V_{S_jS_j"},$
%$\V_{S_jS_j"}$ being the translate of $\A_{S_jS_j"}.$
%Further, usually the blocks $S_j$ would have size one or two,
%even
%when the network has millions of nodes.
%Therefore, it would be easy to build $(\bigoplus \V_{S_jS_j"})^{*}=\bigoplus \V_{S_jS_j"}^{*}.$
%
%We say $\V_{S_j\hat{S}_jS_j"\hat{S}_j"}$ is a \nw{gyrator}
%iff
%$v_{S_j}=-R i_{\hat{S}_j"}; v_{\hat{S}_j}=R i_{S_j"},$ where $R$  is a positive diagonal matrix.
%We denote by $\g ^{S_j\hat{S}_j},$ the gyrator where $R$ is the identity matrix.
%
The set of solutions of  $\mathcal{N}_P\equivd (\Gsp, \K_{SS"})$ may be writen, using the extended
definition of intersection as
$$\V_{SP}\cap (\V^{*}_{SP})_{S"P"}\cap \K_{SS"}=[\V_{SP}\oplus (\V^{*}_{SP})_{S"P"}]\cap \K_{SS"}.$$
We say the multiport is \nw{consistent} iff its set of solutions 
is nonvoid.

\section{Matched and Skewed Composition}
\label{sec:matched}
In this section we introduce an operation between collections of vectors
motivated by the need to capture the relationship between the port voltages and  currents in a multiport.

Let $\Ksp,\Kpq,$ be collections of vectors respectively on $S\uplus P,P\uplus Q,$ with $S,P,Q,$ being pairwise disjoint.
%Further, let $\0_{SP}\in \Ksp, \0_{PQ}\in \K_{PQ}.$

The \nw{matched composition} $\mnw{\mathcal{K}_{SP} \leftrightarrow \mathcal{K}_{PQ}}$ is on $S\uplus Q$ and is defined as follows:
\begin{align*}
 \mathcal{K}_{SP} \leftrightarrow \mathcal{K}_{PQ} 
  &\equivd \{
                 (f_S,g_Q): (f_S,h_P)\in \Ksp, 
(h_P,g_Q)\in \Kpq\}.
%(g|_{(S\setminus Y)} , h|_{(Y \setminus S)}), \textrm{ where } g\in \mathcal{K}_S, h\in \mathcal{K}_Y \textrm{ \& } g|_{(S\cap Y)} = h|_{(S \cap Y)}
%\}.
\end{align*}
Matched composition is referred to as matched sum in \cite{HNarayanan1997}.
It can be regarded as the generalization of composition of maps to composition of relations \cite{narayanan2016}.
%In the special case where $Y\subseteq S$, matched composition is called 
%\nw{generalized minor} operation (generalized minor of $\mathcal{K}_S $
%with respect to $\mathcal{K}_Y$). 

%\subsubsection{Skewed Composition}

The \nw{skewed composition} $\mnw{\mathcal{K}_{SP} \rightleftharpoons \mathcal{K}_{PQ}}$ is on $S\uplus Q$ and is defined as follows:
\begin{align*}
 (\mathcal{K}_{SP} \rightleftharpoons \mathcal{K}_{PQ}) 
  &\equivd \{
                 (f_S,g_Q): (f_S,h_P)\in \Ksp, 
(-h_P,g_Q)\in \Kpq\}.
%(g|_{(S\setminus Y)} , h|_{(Y \setminus S)}), \textrm{ where } g\in \mathcal{K}_S, h\in \mathcal{K}_Y \textrm{ \& } g|_{(S\cap Y)} = h|_{(S \cap Y)}
%\}.
\ \mbox{Note that}
\end{align*}
$$(\mathcal{K}_{SP} \rightleftharpoons \mathcal{K}_{PQ})\ \ \ =\ \ \ \mathcal{K}_{SP} \lrar \mathcal{K}_{(-P)Q}.$$
When $S$, $Y$ are disjoint, both the matched and skewed composition of
$\K_S,\K_Y,$ correspond to the direct sum $\K_S\oplus \K_Y$.
It is clear from the definition of matched composition and that of restriction
and contraction, that
$$\Ksp\lrar \Kpq= (\Ksp\cap \Kpq)\circ SQ;\ \ 
\Ksp\rightleftharpoons \Kpq\ \ =\ \  (\Ksp\cap \K_{(-P)Q})\circ SQ;$$
$$\Ksp\lrar \Kpq = (\Ksp+\K_{(-P)Q})\times SQ;\ \ 
\Ksp\rightleftharpoons \Kpq\ \ =\ \  (\Ksp+ \Kpq)\times SQ.$$
Further, it can be seen that
\\$ \Ks\circ (S-T) =\Ks\lrar \F_T, \Ks\times (S-T)=\Ks\lrar \0_T, T\subseteq S.$
When $\mathcal{K}_{SP}$, $\mathcal{K}_P$ are vector spaces, observe that $(\mathcal{K}_{SP}\leftrightarrow \mathcal{K}_P) = (\mathcal{K}_{SP}\rightleftharpoons \mathcal{K}_P).$
When $S,P,Z,$ are pairwise disjoint,
we have\\
%\begin{align*}
$(\mathcal{K}_{SPZ} \leftrightarrow  \mathcal{K}_{S}) \leftrightarrow  \mathcal{K}_{P} = (\mathcal{K}_{SPZ} \leftrightarrow  \mathcal{K}_{P}) \leftrightarrow  \mathcal{K}_{S} = \mathcal{K}_{SPZ} \leftrightarrow  (\mathcal{K}_{S} \oplus  \mathcal{K}_{P}).$
%\end{align*}
When $ \mathcal{K}_{S}\equivd \0_S, \mathcal{K}_{P}\equivd \F_P,$ the above reduces to
$ \mathcal{K}_{SPZ}\times PZ\circ Z\equaln \mathcal{K}_{SPZ}\circ SZ\times Z.$
%Such an object is called a \nw{minor} of $\mathcal{K}_{SPZ}.$
%In the special case where $Y\subseteq S$, the matched composition $\Ks\lrar \K_Y,$ is called the
%\nw{generalized minor} of $\mathcal{K}_S $
%with respect to $\mathcal{K}_Y$.

The multiport $\mathcal{N}_P\equivd (\Gsp,\K_{SS"})$ would impose a relationship
 between
$v_{P},i_{P"}.$ 
%$([\V^v(\G_{SP})\oplus (\V^i(\G_{SP}))_{S"P"}]\cap \K_{SS"})\circ PP",$ between
%$v_{P},i_{P"}.$ 
This relationship is captured by the  \nw{multiport behaviour}  (port behaviour for short) $\breve{\K}_{PP"}$ at $P,$  of $\N_P,$  defined
by\\
%$$\K_{P'P"}\equiv ((\V_{S'P'}\cap (\V^{\*}_{S'P'})_{S"P"}\cap \K_{S'S"})\circ P'P")_{P'-P"}.$$
$\breve{\K}_{PP"}\equivd [([\V^v(\G_{SP})\oplus (\V^i(\G_{SP}))_{S"P"}]\cap \K_{SS"})\circ PP"]_{P(-P")}= 
([\V^v(\G_{SP})\oplus (\V^i(\G_{SP}))_{S"P"}]\lrar \K_{SS"})_{P(-P")}
$\\$= [\V^v(\G_{SP})\oplus (\V^i(\G_{SP}))_{S"(-P")}]\lrar \K_{SS"}.$
When the device characteristic of $\mathcal{N}_P$ is affine, its {port behaviour} $\breve{\K}_{PP"}$ at $P$ would be
affine if it were not void.

Note that,
if the multiport
is a single port edge in parallel with a positive resistor $R,$\\
$([\V_{SP}\oplus (\V^{*}_{SP})_{S"(-P")}]\cap \K_{SS"})\circ PP"$
would be the solution of $v_{P}= -Ri_{P"}.$
But then $\breve{\K}_{PP"},$ as defined, would be the solution of $v_{P}= Ri_{P"}.$

Let the multiports $\N_{RP},{\N}_{\tilde{P}Q}$ be on graphs $\G_{RSP},\G_{\tilde{P}MQ}$ respectively, with the primed and double primed sets obtained from 
$R,S,P,\tilde{P},M,Q,$ being pairwise disjoint,
%$S',T',P',\tilde{P}',Q',W',S",T",P",\tilde{P}",Q",W"$ pairwise disjoint,
and let them have device characteristics $\K^{S},{\K}^{M}$ respectively.
Let $\K^{P\tilde{P}}$ denote a collection of vectors $\K^{P\tilde{P}}_{P\tilde{P}P"\tilde{P}"}.$
The multiport $\mnw{[\N_{RP}\oplus {\N}_{\tilde{P}Q}]\cap \K^{P\tilde{P}}},$ 
with ports $R\uplus Q$ obtained by \nw{connecting $\N_{RP},{\N}_{\tilde{P}Q}$
through $\K^{P\tilde{P}}$},
is on graph $\G_{RSP}\oplus\G_{\tilde{P}MQ}$ (the graph obtained by putting 
$\G_{RSP},\G_{\tilde{P}MQ}$ together with no common nodes) with device characteristic
\\$\K^{S}\oplus {\K}^{M}\oplus  \K^{P\tilde{P}}.$ 
%(see Figure \ref{fig:networkmultiportconnection}).
When $R,Q$ are void, $[\N_{RP}\oplus {\N}_{\tilde{P}Q}]\cap \K^{P\tilde{P}}$ would reduce to $[\N_{P}\oplus {\N}_{\tilde{P}}]\cap \K^{P\tilde{P}}$  and would be a network without ports.
In this case we say \nw{the multiport $\N_{P}$ is terminated by ${\N}_{\tilde{P}}$ through $\K^{P\tilde{P}}.$}\\
This network is on  graph $\G_{SP}\oplus \G_{\tilde{P}M}$ with device
characteristic $\K^S_{SS"}
\oplus  {\K}^M_{MM"}\oplus \K^{P\tilde{P}}_{PP"\tilde{P}\tilde{P}"}.$

The following result is useful for relating the port behaviour of a multiport 
with internal sources to that of the source zero version of the multiport.
Its routine proof is omitted.
\begin{theorem}
\label{thm:IIT2}
Let $\Asp,\Apq$ be affine spaces on $S\uplus P,P\uplus Q,$ where $S,P,Q$ 
are pairwise disjoint sets. Let $\Vsp,\Vpq$ respectively, be the vector 
space translates of $\Asp,\Apq.$ Let $\Asp\lrar \Apq$ be \nw{nonvoid} and let\\
$\alpha_{SQ}\in \Asp\lrar \Apq.$ Then,
 $\Asp\lrar \Apq = \alpha_{SQ}+(\Vsp\lrar \Vpq).$
\end{theorem}
The following result is an  immediate consequence of Theorem \ref{thm:IIT2}.
\begin{theorem}
\label{thm:translatemultiport}
 Let $\N^1_P,{\N}^2_P$ be multiports on the same graph $\G_{SP}$
but with device characteristics \\
$\A_{SS"}, {\V}_{SS"},$ respectively where ${\V}_{SS"},$
is the vector space translate of the affine space $\A_{SS"}.$
Let $\N^1_P,{\N}^2_P$ have 
port behaviours $\breve{\A}^1_{PP"}, \breve{\V}^2_{PP"},$ respectively.
If $\breve{\A}^1_{PP"}\ne \emptyset ,$ then \\$\breve{\V}^2_{PP"}=
([\V^v({\G_{SP}})\oplus (\V^i({\G_{SP}}))_{S"P"}]\lrar {\V}_{SS"})_{P(-P")},$
is the vector space translate of $\breve{\A}^1_{PP"}.$
\end{theorem}
%Other applications of Theorem \ref{thm:IITlinear} may be found in 
%\cite{HNarayanan1997,HNPS2013, narayanan2016}.

%\section{Implicit Duality Theorem and its application to multiports}
%\label{sec:idt}
%In this section we state the second basic result of ILA and present some 
%applications.
Implicit Duality Theorem, given below, is a part of network theory folklore.
However, its applications are insufficiently emphasized in the literature.
Proofs and applications may be found in \cite{HNarayananadjoint,HNarayanan1997,schaft1999,narayanan2016}.
A version in the context of Pontryagin Duality is available in \cite{forney2004}.
\begin{theorem}
\label{thm:idt0}
{\bf Implicit Duality Theorem} 
Let $\Vsp, \Vpq$ be vector spaces respectively on $S\uplus P,P\uplus Q,$ with $S,P,Q,$ being pairwise disjoint. We then have,
$(\mathcal{V}_{SP}\leftrightarrow \mathcal{V}_{PQ})^* 
\ \equaln\ (\mathcal{V}_{SP}^* \rightleftharpoons \mathcal{V}_{PQ}^*) 
.$ In particular,
$(\mathcal{V}_{SP}\leftrightarrow \mathcal{V}_{P})^* \ \equaln\ \mathcal{V}_{SP}^* \leftrightarrow \mathcal{V}_{P}^*
.$
\end{theorem}
\begin{proof}
It can be shown that $(\V_X+\V_Y)^{*}=(\V^{*}_X\cap \V_Y^{*})$
and, using Theorem \ref{thm:perpperp}, that $(\V_X\cap\V_Y)^{*}=(\V^{*}_X+\V_Y^{*}).$
We have, using Theorem \ref{thm:dotcrossidentity},
$(\Vsp\lrar \Vpq)^{*}= [(\Vsp\cap \Vpq)\circ SQ]^{*}=
[(\Vsp^{*}+ \Vpq^{*})\times SQ]= (\Vsp^{*}\rightleftharpoons\Vpq^{*}).$
For any vector space $\V_X,$ we have $\V_X= \V_{(-X)}.$ 
Therefore $(\mathcal{V}_{SP}\leftrightarrow \mathcal{V}_{P})^* = (\mathcal{V}_{SP}^* \rightleftharpoons \mathcal{V}_{P}^*)
= \mathcal{V}_{SP}^* \lrar \mathcal{V}_{P}^*
.$
\end{proof}
%A proof of a general version of Theorem \ref{thm:idt0}, is given in the
%appendix.
An illustration of   the use of Theorem \ref{thm:idt0} is provided in the 
next subsection (Theorem \ref{thm:adjointmultiport}). 
\subsection{Adjoint of an affine space and the adjoint multiport}
\label{sec:adjoint}
%First we need the definitions of  some special characteristics and relationships between them.  \\
We say $\V_{S_jS_j"},\hat{\V}_{S_jS_j"}$ are \nw{orthogonal duals}
of each other iff $\hat{\V}_{S_jS_j"}= \V^{*}_{S_jS_j"}.$
%and denote the orthogonal duals of $\V_{S_j'S_j"}$ by $\V^{dual}_{S_j'S_j"}.$
By Theorem \ref{thm:perpperp},\\ $(\V^{*}_{S_jS_j"})^{*}=\V_{S_jS_j"}.$
Let $\A_{S_jS_j"}$ be an affine space with $\V_{S_jS_j"}$ as its 
vector space translate.
We say $\hat{\V}_{S_jS_j"}$ is the  \nw{adjoint}
of $\A_{S_jS_j"},$ 
denoted by ${\V}^{adj}_{S_jS_j"},$
iff $\hat{\V}_{S_jS_j"}= (\V^{*}_{S_jS_j"})_{(-S_j")S_j}.$
It is easy to see that $(\oplus \A_{S_jS_j"})^{adj}= (\oplus \V_{S_jS_j"}^{adj}).$
It is clear that 
${\V}_{S_jS_j"}= (\hat{\V}^{*}_{S_jS_j"})_{(-S_j")S_j}= ({\V}^{adj}_{S_jS_j"})^{adj}.$
%It is easy to see that $(\oplus \A_{S_j'S_j"})^{adj}= (\oplus \V_{S_j'S_j"}^{adj}).$
\begin{example}
\label{eg:egadjoint1}
Let ${\A}_{SS"}$ be an affine device characteristic.
Suppose its vector space translate ${\V}_{SS"}$ is defined by (is the solution space of) the hybrid equations
\begin{align}
\label{eqn:egadjoint1}
%\ppmatrix{i_{S_1}\\v_{S_2}}-\ppmatrix{g_{11}&h_{12}\\h_{21}&r_{22}}\ppmatrix{v_{S_1}\\i_{S_2}}&\ \ \ \ \ \ \ \ \ \ \ \ = &\ppmatrix{0\\0},\\  &\mbox{equivalently}&\\
\ppmatrix{I&0&-g_{11}&-h_{12}\\0&I&-h_{21}&-r_{22}}\ppmatrix{i_{S_1"}\\v_{S_2}\\v_{S_1}\\i_{S_2"}} =  \ppmatrix{0\\0},
\end{align}
where $S$ is partitioned into $S_1,S_2.$
Then ${\V}^*_{SS"}$ is the complementary orthogonal space defined by
\begin{align}
\label{eqn:egadjoint2}
%\ppmatrix{i_{S_1}\\v_{S_2}}-\ppmatrix{g_{11}&h_{12}\\h_{21}&r_{22}}\ppmatrix{v_{S_1}\\i_{S_2}}&\ \ \ \ \ \ \ \ \ \ \ \ = &\ppmatrix{0\\0},\\  &\mbox{equivalently}&\\
\ppmatrix{g^*_{11}&h^*_{21}&I&0\\h^*_{12}&r^*_{22}&0&I}\ppmatrix{\ \  \tilde{i}_{S_1"}\\\tilde{v}_{S_2}\\\tilde{v}_{S_1}\\\tilde{i}_{S_2"}} =  \ppmatrix{0\\0},
\end{align}
and ${\V}^{adj}_{SS"}\equivd  ({\V}^*_{SS"})_{(-S")S}$ is obtained by 
replacing the $\tilde{v}$ variables by $-\hat{i}$ variables and 
the $\tilde{i}$ variables by $\hat{v}$ variables.
It can be seen that ${\V}^{adj}_{SS"}$ is defined by
\begin{align}
\label{eqn:egadjoint3}
%\ppmatrix{i_{S_1}\\v_{S_2}}-\ppmatrix{g_{11}&h_{12}\\h_{21}&r_{22}}\ppmatrix{v_{S_1}\\i_{S_2}}&\ \ \ \ \ \ \ \ \ \ \ \ = &\ppmatrix{0\\0},\\  &\mbox{equivalently}&\\
\ppmatrix{g^*_{11}&h^*_{21}&I&0\\h^*_{12}&r^*_{22}&0&I}\ppmatrix{\hat{v}_{S_1}\\ - \hat{i}_{S_2"}\\ - \hat{i}_{S_1"}\\\hat{v}_{S_2}}=\ppmatrix{0\\0},\mbox{i.e.,} \ppmatrix{-g^*_{11}&h^*_{21}&I&0\\h^*_{12}&-r^*_{22}&0&I}\ppmatrix{\hat{v}_{S_1}\\ \hat{i}_{S_2"}\\ \hat{i}_{S_1"}\\\hat{v}_{S_2}} =  \ppmatrix{0\\0}.
\end{align}
The individual devices which are present  will  usually have very few  ports.
 Building their adjoints is therefore computationally inexpensive. 
%If we build another multiport by retaining the same graph but replacing these devices by their adjoints, at the ports we will see a behaviour adjoint to the original multiport 
%behaviour (see Theorem \ref{thm:adjointmultiport}, below).
%Some standard devices and their adjoints are given below.
\\
1. Let $v=0,i=0$ be a nullator. The adjoint will have its voltage and current 
unconstrained and is therefore a norator.
\\
2. Let $v_S= Zi_{S"}$ be a (multiport)   impedance. 
%$R$ being diagonal,
%with nonzero diagonal entries.
%The vector space translate of the characteristic is the solution space
%of $v= Ri.$ 
The adjoint has the characteristic $\hat{v}_S=Z^*\hat{i}_{S"}.$
%Thus the  resistor is self adjoint or reciprocal.
%This would be true even if $R$ is a symmetric matrix of full rank.
\\3. Consider the current controlled voltage source (CCVS)
and the  voltage  controlled current source (VCCS)
shown below.
\begin{align}
\label{eqn:controlledsourcesp}
%\ppmatrix{1&-h_{12}&-r_{11}&0\\
%0&-g_{22}&-h_{21}&1}\ppmatrix{v_1\\v_2\\i_1\\i_2}=\ppmatrix{s_1\\s_2};
\ppmatrix{1&0&0&0\\
0&1&-r&0}\ppmatrix{v_1\\v_2\\i_1\\i_2}=\ppmatrix{0\\0};
\ppmatrix{0&0&1&0\\
-g&0&0&1}\ppmatrix{v_1\\v_2\\i_1\\i_2}=\ppmatrix{0\\0};
\end{align}
The adjoints are respectively (built by first building the orthogonal dual
of the source zero characteristic,
interchanging current and voltage variables and  changing the sign of the current variables)
\begin{align}
\label{eqn:controlledsources2p}
%\ppmatrix{1&-h_{12}&-r_{11}&0\\
%0&-g_{22}&-h_{21}&1}\ppmatrix{v_1\\v_2\\i_1\\i_2}=\ppmatrix{s_1\\s_2};
\ppmatrix{0&-r&1&0\\
0&0&0&1}\ppmatrix{\hat{i}_1\\\hat{i}_2\\\hat{v}_1\\\hat{v}_2}=\ppmatrix{0\\0};
\ppmatrix{1&0&0&-g\\
0&1&0&0}\ppmatrix{\hat{i}_1\\\hat{i}_2\\\hat{v}_1\\\hat{v}_2}=\ppmatrix{0\\0}.
\end{align}
Thus the adjoints of  CCVS, VCCS remain  CCVS,VCCS respectively,
with the same parameters $r,g$ respectively,
but the direction of control which was originally from port $1$ to
port $2$ is now  from port $2$ to
port $1.$
\\
4. Next consider  current  controlled current source (CCCS)
and the  voltage controlled voltage source (VCVS)
shown below.
\begin{align}
\label{eqn:controlledsources3p}
%\ppmatrix{1&-h_{12}&-r_{11}&0\\
%0&-g_{22}&-h_{21}&1}\ppmatrix{v_1\\v_2\\i_1\\i_2}=\ppmatrix{s_1\\s_2};
\ppmatrix{1&0&0&0\\
0&0&-\alpha &1}\ppmatrix{v_1\\v_2\\i_1\\i_2}=\ppmatrix{0\\0};
\ppmatrix{0&0&1&0\\
-\beta &1&0&0}\ppmatrix{v_1\\v_2\\i_1\\i_2}=\ppmatrix{0\\0}.
\end{align}
The adjoints are respectively,
\begin{align}
\label{eqn:controlledsources4p}
%\ppmatrix{1&-h_{12}&-r_{11}&0\\
%0&-g_{22}&-h_{21}&1}\ppmatrix{v_1\\v_2\\i_1\\i_2}=\ppmatrix{s_1\\s_2};
\ppmatrix{1&0&0&0\\
0&0&1&\alpha}\ppmatrix{\hat{i}_1\\\hat{i}_2\\\hat{v}_1\\\hat{v}_2}=\ppmatrix{0\\0};
\ppmatrix{0&0&1&0\\
1&\beta &0&0}\ppmatrix{\hat{i}_1\\\hat{i}_2\\\hat{v}_1\\\hat{v}_2}=\ppmatrix{0\\0}.
\end{align}
Thus the adjoints of  CCCS, VCVS are source zero VCVS,CCCS respectively,
but the direction of control which was originally from port $1$ to
port $2$ is now from port $2$ to
port $1.$ Also the current gain factor is now the negative of the voltage gain factor
in the CCCS to VCVS case and vice versa in VCVS to CCCS case.
\end{example}
Let the multiport $\mnw{{\N}_P}$ be on graph $\G_{SP}$
with device characteristic
$\A_{SS"}=x_{SS"}+{\V}^{}_{SS"}$ on $S.$
The  multiport $\mnw{{\N}^{hom}_P}$ is on graph $\G_{SP}$
but has device characteristic ${\V}^{}_{SS"}.$
The \nw{adjoint} $\N^{adj}_P$ of ${\N}_P$ as well as of ${\N}^{hom}_P$ 
is on graph $\G_{SP}$
but has device characteristic ${\V}^{adj}_{SS"}.$
%is obtained from $\N_P$ by replacing the device characterisitic\\
% $\K_{SS"}=x_{S'S"}+{\V}^{}_{S'S"}$ on $S,$
%by its `homogeneous' version, i.e., its translate ${\V}^{}_{S'S"}.$\\
%%with $P$ being a set of norators.\\
%The \nw{adjoint} $\N^{adj}_P$ of ${\N}_P$ as well as ${\N}^{hom}_P$ 
%
%
%is
%obtained from the latter by
%building the adjoint\\ of $\V_{S'P'}\oplus (\V^{\*}_{S'P'})_{S"P"}$
%and of ${\V}^{}_{S'S"},$
% as follows.\\
%%In place of $\V_{S'P'}\oplus (\V^{\*}_{S'P'})_{S"P"}$ in ${\N}^{hom}_P,$
%In place of $\V_{S'P'}\oplus (\V^{\*}_{S'P'})_{S"P"}$ in ${\N}^{hom}_P,$
%we have, in $\N^{adj}_P,$\\
%$(\V_{S'P'}\oplus (\V^{\*}_{S'P'})_{S"P"})^{adj}\equiv (\V_{S'P'}\oplus (\V^{\perp}_{S'P'})_{S"P"})^{\perp}_{-S"(-P")S'P'}=(\V_{S'P'}^{\perp}\oplus ((\V_{S'P'}^{\perp})_{S"P"}){^\perp})_{-S"(-P")S'P'}$\\$= (\V^{\perp}_{S'P'})_{S"P"}\oplus \V_{S'P'}= \V_{S'P'}\oplus (\V^{\perp}_{S'P'})_{S"P"}$\\ and
%%= (\V_{S'P'})_{S"P"}\oplus \V^{\perp}_{S'P'}$\\ and 
%in place of device characteristic
%% on $S$ being 
%${\V}^{}_{S'S"}$
% in ${\N}^{hom}_P,$
% we have, in $\N^{adj}_P,$\\
% ${\V}^{adj}_{S'S"}\equiv ({\V}^{\perp}_{S'S"})_{-S"S'},$
%with $P$ being a set of norators.

Let the solution set of $\N_P$
be
$[\V_{SP}\oplus (\V^{*}_{SP})_{S"P"}]\cap \A_{SS"}.$
Then the solution set of  ${\N}^{hom}_P$
would be
$[\V_{SP}\oplus (\V^{*}_{SP})_{S"P"}]\cap {\V}^{}_{SS"},
$
and that of  $\N^{adj}_P$
would be
$[\V_{SP}\oplus (\V^{*}_{SP})_{S"P"}]\cap {\V}^{adj}_{SS"}.$

The port behaviour of $\N_P$
would be $[\V_{SP}\oplus (\V^{*}_{SP})_{S"(-P")}]\lrar \A_{SS"},$
that of 
${\N}^{hom}_P$
would be\\
$[\V_{SP}\oplus (\V^{*}_{SP})_{S"(-P")}]\lrar {\V}^{}_{SS"},
$
and that of  $\N^{adj}_P$
would be
$[\V_{SP}\oplus (\V^{*}_{SP})_{S"(-P")}]\lrar {\V}^{adj}_{SS"}.$

%If $\K_{S'S"},$ has a vector space translate $\V_{S'S"},$
%of dimension $|S|,$ it is clear that $\V^{adj}_{S'S"}$
%has dimension $2|S|-r(\V_{S'S"})=|S|.$ We
% therefore have the following result, using Lemma \ref{}.

We now have the following basic result on linear multiports \cite{HNarayanan1997,HNarayananadjoint}.
It essentially states that, 
$\N^{hom}_P$ and $\N^{adj}_P$ have adjoint port behaviours.
\begin{theorem}
\label{thm:adjointmultiport}
%\begin{enumerate}
Let $\N^1_P,{\N}^2_P$ be multiports on the same graph $\G_{SP}$
but with device characteristics \\
$\V^1_{SS"}, {\V}^2_{SS"},$ respectively
and port behaviours $\breve{\V}^1_{PP"}, \breve{\V}^2_{PP"},$ respectively.
Then if $\V^1_{SS"}, {\V}^2_{SS"},$ are adjoints  of each other
so are $\breve{\V}^1_{PP"}, \breve{\V}^2_{PP"},$ adjoints  of each other.
%\item If the port behaviour $\breve{\V}_{P'P"}$ of $\N_P$ is proper,
%then so is the port behaviour $(\breve{\V}_{P'P"})^{adj}$ of $\N^{adj}_P.$
%\item If $\N_P$ has a device characteristic $\V_{S'S"}$ that is reciprocal
%(self-Dirac dual) then its port behaviour $\breve{\V}_{P'P"}$ is also
%reciprocal
%(self-Dirac dual). Further the port behaviour has dimension $|P|.$
%\item If $\N_P$ has a device characteristic $\V_{S'S"}$ that is an ideal transformer,
%then its port behaviour $\breve{\V}_{P'P"}$ is also
%an ideal transformer.
%\end{enumerate}
\end{theorem}
\begin{proof}
 Let $\V_{SP}\equivd \V^v(\G_{SP}).$
By  Theorem \ref{thm:tellegen}, we have $((\V^v(\G_{SP}))^{*})_{S"P"}=\V^i(\G_{SP}).$\\
Therefore $(\V_{SP}^{*})_{S"P"}=\V^i(\G_{SP}).$
Let $\V_{SPS"P"}\equivd\V_{SP}\oplus (\V_{SP}^{*})_{S"P"}.$
It is clear that $\V_{SPS"P"}^{adj}$\\$\equivd (\V_{SPS"P"}^{*})_{(-S")(-P")SP}=
\V_{SPS"P"}.$
%Let $\V^1_{P'P"}\equivd(\V_{S'P'}\oplus ((\V_{S'P'})^{\*})_{S"P"})\lrar \V^1_{S'S"}.$ and let ${\V}^2_{P'P"}=(\V_{S'P'}\oplus ((\V_{S'P'})^{\perp})_{S"P"})\lrar {\V}^2_{S'S"}.$
%
%It is clear that $\V_{S'P'}\oplus ((\V_{S'P'})^{\perp})_{S"P"}$
%is self-adjoint. We are given that $\V^1_{S'S"},{\V}^2_{S'S"}$ are adjoints.
 Let $ \V^1_{PP"}\equivd \V_{SPS"P"}\lrar \V^1_{SS"}$
and let $ {\V^2}_{PP"}\equivd {\V}_{SPS"P"}\lrar {\V}^2_{SS"}.$
\\
We then have,
$ (\V^1_{PP"})^{adj}\equivd (\V^1_{PP"})^{*}_{(-P")P}= (\V_{SPS"P"}\lrar \V_{SS"})^{*}_{(-P")P}= 
(\V_{SPS"P"}^{*}\lrar \V_{SS"}^{*})_{(-P")P}$\\
$=(\V_{SPS"P"}^{*})_{S(-P")S"P}\lrar \V_{SS"}^{*}$
$=(\V_{SPS"P"}^{*})_{(-S")(-P")SP}\lrar  (\V_{SS"}^{*})_{(-S")S}
%$=(\V_{SPS"P"}^{*})_{S"(-P")SP}\lrar (\V_{SS"}^{*})_{S"S}=
%$\\$=(\V_{SPS"P"}^{*})_{(-S")(-P")SP}\lrar (\V_{SS"}^{*})_{(-S")S}
={\V}_{SPS"P"}\lrar {\V^2}_{SS"}={\V^2}_{PP"} .$
%
%==============================================================
%
%$= (\V_{SP}^{*}\oplus (\V_{SP})_{S"P"})_{S"(-P")SP}\lrar (\V_{SS"}^{*})_{S"S}
%= (\V_{SP}^{*}\oplus (\V_{SP})_{S"P"})_{(-S")(-P")SP}\lrar (\V_{SS"}^{*})_{(-S")S}
%=(\V_{SPS"P"}^{*})_{(-S")(-P")SP}\lrar (\V_{SS"}^{*})_{-S"S}= {\V}_{SPS"P"}\lrar {\V^2}_{SS"}={\V^2}_{PP"} .$
Thus  $\V^1_{PP"},{\V}^2_{PP"}$ are adjoints and therefore
$\breve{\V}^1_{PP"}\equivd (\V^1_{PP"})_{P(-P")}$ and
$\breve{\V}^2_{PP"}\equivd ({\V}^2_{PP"})_{P(-P")}$ are adjoints.
\end{proof}
\begin{corollary}
\label{cor:adjointmultiport}
Let $\N_P\equivd (\Gsp,\A_{SS"}),$ with $\V_{SS"}$ as the vector space 
translate of $\A_{SS"}.$
Suppose $\breve{\A}_{PP"}$ is the (nonvoid) port behaviour of $\N_P$
with $\breve{\V}_{PP"}$ as its  vector space
translate. 
Then $\N^{adj}_P\equivd (\Gsp,\V^{adj}_{SS"}),$ has the port behaviour 
$\breve{\V}^{adj}_{PP"}.$ 
\end{corollary}
\begin{proof}
%Let $\V_{SPS"P"}\equivd \V^v(\G_{SP})\oplus (\V^i(\G_{SP}))_{S"P"}
%.$
By Theorem \ref{thm:translatemultiport}, 
port behaviour of $\N^{hom}_P\equivd (\Gsp,\V_{SS"}),$
is the vector space translate of $\breve{\A}_{PP"}.$
%
%$\V_{SPS"P"}\lrar \V_{SS"}$ is the vector space associate of $\V_{SPS"P"}\lrar \A_{SS"}= (\breve{\A}_{PP"})_{P(-P")}.$  
%Therefore the port behaviour of $\N^{hom}_P\equivd (\Gsp,\V_{SS"}),$
%s the vector space associate of $\breve{\A}_{PP"}.$
By definition, the adjoint of $\N_P$ and $\N^{hom}_P$ are the same.
The result now follows from Theorem \ref{thm:adjointmultiport}, noting 
that $(\breve{\V}^{adj}_{PP"})_{P(-P")}$ is the  
 adjoint of $(\breve{\A}_{PP"})_{P(-P")}.$
\end{proof}
\begin{remark}
\label{rem:adj1}
1. Suppose the original port behaviour $\A_{PP"}$ is defined by
$v_P=Zi_P+E.$ Its vector space translate $\V_{PP"}$ is defined by
$v_P=Zi_P.$ i.e., by \begin{align} (I|-Z)\ppmatrix{v_P\\i_{P"}}=0.\end{align}
$\V^{adj}_{PP"}$ is defined by \begin{align} (I|-Z^*)\ppmatrix{v_P\\i_{P"}}=0,\end{align} i.e., by $v_P=Z^*i_P.$
Now let every device in the multiport have the form $v_{S_j}=Z_ji_{S_{j}"}+E_j$
and, further, at the ports let the behaviour be $v_P=Zi_P+E.$ 
Theorem \ref{thm:adjointmultiport} implies that if every device is replaced by
$v_{S_j}=Z^*_ji_{S_{j}"},$ at the ports the behaviour would be $v_P=Z^*i_P.$
But, as the theorem indicates, the idea of the adjoint works well even if we have only a relationship
of the kind $Bv_P-Qi_P=s.$ 

2. The devices in a multiport have few device ports. So building their
adjoints is easy. 
Therefore $\N^{adj}_P$ can be built essentially in linear time.
Thus we can ``implicitly'' build, essentially in linear time, the adjoint of the port behaviour of $\N_P.$ We use ``implicitly'' because both the 
port behaviours are available not as explicit equations, but in terms of multiports.
This fact is exploited subsequently to generalize Thevenin-Norton and maximum 
power transfer theorems.
\end{remark}
%

%\section{Generalizing Thevenin-Norton: computing port behaviour solving special circuits}
\section{Generalization of Thevenin-Norton Theorem}
\label{sec:computingbehaviour}
The characteristic feature of the Thevenin-Norton Theorem is that 
it is in terms of repeated solution of networks obtained by suitable 
termination of a given multiport. These networks are assumed to have 
 unique solutions.
Our generalization of the Thevenin-Norton theorem is in accordance 
with this feature and requires the notion 
of a rigid multiport. If a multiport is not rigid, it cannot be a part 
of a network with unique solution and therefore we cannot use a conventional simulator to solve a network
 of which the multiport is a part.
A rigid multiport 
can permit 
any nonvoid port behaviour at its ports. 
 {\it Therefore our generalization of the Thevenin-Norton theorem in terms of rigid multiports is the best that is possible if 
conventional circuit simulators are to be used to compute the port behaviour.}
Our technique is to terminate a rigid multiport by its adjoint through 
a gyrator. We therefore need to prove that the adjoint of a rigid multiport is rigid. We need the development in the next subsection for this purpose.
\subsection{Rigid  Multiports}
\begin{definition}
\label{def:regular}
Let multiport $\N_P\equivd (\G_{SP}, \A_{S'S"}), $
%be on graph $\G_{SP}$ and device characteristic 
where
 $\A_{S'S"}
=\alpha_{S'S"}+\V_{S'S"}.$\\
The multiport $\N_P$ is said to be \nw{rigid
} iff every  multiport $\hat{\N}_P\equivd (\G_{SP}, \hat{\A}_{S'S"}), $
%on graph $\G_{SP}$ and device characteristic
 where $\hat{\A}_{S'S"}=\hat{\alpha}_{S'S"}+\V_{S'S"},$
has a non void set of
solutions
and has a unique solution corresponding to every vector in its multiport behaviour.
\\Let $\A_{AB},\A_{B}$ be affine spaces on sets $A\uplus B, B,$ respectively, $A,B,$ disjoint.
Further, let $\A_{AB},\A_{B}$ have vector  space translates $\V_{AB},\V_{B},$
respectively.
\\
We say the pair $\{\A_{AB},\A_{B}\}$ has the \nw{full sum property} iff
$\V_{AB}\circ B+\V_{B}=\F_B.$\\
We say the pair $\{\A_{AB},\A_{B}\}$ has the \nw{zero intersection property} iff
$\V_{AB}\times B\cap \V_{B}=\0_B.$\\
We say that the pair $\{\A_{AB},\A_{B}\}$ is \nw{rigid,
} iff  it has the full sum property and the zero intersection property.
\end{definition}
Multiport rigidity reduces to affine space pair rigidity and the adjoint of a rigid multiport is also rigid. 
\begin{theorem}
\label{thm:regularrecursive}
Let $\{\A_{AB},\A_{B}\}$ be a rigid
 pair and let $\V_{AB},\V_{B}$ be the vector space translates of $\A_{AB},\A_{B},$ respectively. Then
\begin{enumerate}
\item The full sum property  of $\{\A_{AB},\A_{B}\}$ is equivalent to
 $\hat{\A}_{AB}\lrar \hat{\A}_{B}$ being nonvoid, whenever
$\V_{AB},\V_{B}$ are the vector space translates of $\hat{\A}_{AB}, \hat{\A}_{B},$ respectively.
Further,
 when the full sum property holds for $\{\A_{AB},\A_{B}\},$  $\hat{\A}_{AB}\lrar \hat{\A}_{B}$  has  vector space translate $\V_{AB}\lrar \V_{B}.$
\item The zero intersection  property of $\{\A_{AB},\A_{B}\}$ is equivalent to
the statement that , \\if $f_A\in \A_{AB}\lrar \A_{B}$ and $(f_A,f_B), (f_A,f'_B)\in \A_{AB}\cap \A_{B},$
then $f_B=f'_B.$
\item The pair $\{\A_{AB},\A_{B}\}$ has the zero intersection (full sum) property iff $\{\V^{*}_{AB},\V^{*}_{B}\}$ has the full sum (zero intersection) property.
Therefore
 $\{\A_{AB},\A_{B}\}$ is {rigid
} iff
$\{\V^{*}_{AB},\V^{*}_{B}\}$ is rigid.
\item A multiport $\N_P$ is rigid iff  $\N^{hom}_P$ and $\N^{adj}_P$
are rigid.
\item A network $\N\equivd (\G_S,\A_{SS"}),$ where $\A_{SS"}$ is proper, has 
a unique solution iff $\N^{hom}\equivd(\G_S,\V_{SS"}),$ 
where $\V_{SS"}$ is the vector space translate of $\A_{SS"},$
has a unique solution.

\end{enumerate}
\end{theorem}
\begin{proof}
1.
Let $\hat{\A}_{AB}=(\alpha_A,\alpha_B)+\V_{AB}, \hat{\A}_{B}=\beta_{B}+\V_{B}.$
\\By the definition of matched composition, $\hat{\A}_{AB}\lrar \hat{\A}_{B}$ is nonvoid iff $\hat{\A}_{AB}\cap \hat{\A}_{B}$ is nonvoid,\\ i.e., iff $\hat{\A}_{AB}\circ B\ \cap\  \hat{\A}_{B}$ is nonvoid,\\
i.e., iff there exist $\lambda_B\in \V_{AB}\circ B, \sigma_B\in \V_{B},$
such that $\alpha_B+\lambda_B=\beta_B+\sigma_B,$ i.e., 
%i.e., iff there exist $\lambda_B\in \V_{AB}\circ B, \sigma_B\in \V_{B},$
such that $\alpha_B-\beta_B=\sigma_B-\lambda_B.$\\
Clearly when $\V_{AB}\circ B+\V_{B}=\F_B,$
there exist $\lambda_B\in \V_{AB}\circ B, \sigma_B\in \V_{B},$
such that $\alpha_B-\beta_B=\sigma_B-\lambda_B$
 so that $\hat{\A}_{AB}\lrar \hat{\A}_{B}$ is nonvoid.
\\
On the other hand, if $\V_{AB}\circ B+\V_{B}\ne \F_B,$
there exist $\alpha_B,\beta_B$ such that $\alpha_B-\beta_B \notin\V_{AB}\circ B+\V_{B},$ so that $\hat{\A}_{AB}\cap \hat{\A}_{B}$ and therefore
$\hat{\A}_{AB}\lrar \hat{\A}_{B}$ is void.\\
By Theorem \ref{thm:IIT2}, if $\hat{\A}_{AB}\lrar \hat{\A}_{B}$ is nonvoid, its vector space translate is $\V_{AB}\lrar \V_{B}.$

2.
%$(\V_{AB}\cap \V_{BC})\times B=\0_B.$\\
Let $(\V_{AB}\times B)\cap \V_{B}=\0_B.$
If $f_A\in \A_{AB}\lrar \A_{B}$ and $(f_A,f_B), (f_A,f'_B)\in \A_{AB}\cap \A_{B},$
then\\ $(0_A,(f_B-f'_B))\in \V_{AB},$ and  similarly $(f_B-f'_B)\in \V_{B},$ so that $(f_B-f'_B)\in (\V_{AB}\times B)\cap \V_{B}=\0_B.$
\\
Next suppose whenever $f_A\in \A_{AB}\lrar \A_{B}$ and $(f_A,f_B), (f_A,f'_B)\in \A_{AB}\cap \A_{B},$       we have $f_B=f'_B.$
Suppose $(\V_{AB}\times B) \cap \V_{B}\ne \0_B.$
Let $g_B\in (\V_{AB}\times B) \cap \V_{B}$ and let $g_B\ne 0_B.$
\\If $f_A\in \A_{AB}\lrar \A_{B},$ then there exists $f_B$ such that  $(f_A,f_B) \in \A_{AB}$ and
$f_B \in  \A_{B}.$
Clearly\\ $(f_A,f_B)+(0_A,g_B)=(f_A,f_B+g_B) \in \A_{AB}$ and
 $(f_B + g_B) \in  \A_{B},$
so that we must have $(f_A,f_B+g_B) \in \A_{AB}\cap \A_{B}.$
But this means $f_B=f_B+g_B,$
a contradiction.
We conclude that $(\V_{AB}\times B)\cap \V_{B}=\0_B.$

 3. We have, 
$(\V_{AB}\circ B+\V_{B})^{*}=\V_{AB}^{*}\times B\ \cap \ \V_{B}^{*}=\F^{*}_B= \0_B,$
and\\
$(\V_{AB}\times B\ \cap\ \V_{B})^{*}=\V_{AB}^{*}\circ B+ \V_{B}^{*}=\0^{*}_B= \F_B.$

4. By part 1 and 2 above, the rigidity of $\N_P\equivd (\Gsp,\A_{SS"})$ 
is equivalent to the rigidity of\\ $(\V^v(\Gsp)\oplus (\V^i(\Gsp))_{S"P"},\V_{SS"}).$ By part 3, the rigidity of the latter is equivalent to the rigidity of 
\\$((\V^v(\Gsp)\oplus (\V^i(\Gsp))_{S"P"})^*, \V^*_{SS"}),$
i.e., to the rigidity of $((\V^i(\Gsp))_{SP}\oplus (\V^v(\Gsp))_{S"P"}, \V^*_{SS"})$
\\(using Theorem \ref{thm:tellegen}), i.e., to the rigidity of 
$\N_P^{adj}\equivd (\V^v(\Gsp)\oplus (\V^i(\Gsp))_{S"P"}, (\V^*_{SS"})_{(-S")S}).$

5. Let $\N\equivd (\G_S,\A_{SS"}),\N^{hom}\equivd (\G_S,\V_{SS"}).$
%The linear network has a solution always, iff it satisfies the 
%full sum property, i.e., iff $r(\V^v(\G_S)\oplus (\V^i(\G_S))_{S"})+\V_{SS"})=2|S|.$
We have  $r(\V^v(\G_S)\oplus (\V^i(\G_S))_{S"})=|S|$ (using Theorems \ref{thm:perpperp}, \ref{thm:tellegen}) and $r(\V_{SS"})=|S|.$
Both $\N$ as well as $\N^{hom}$ have the same $2|S|\times 2|S|$ coefficient matrix of the network
 equations. The networks have  unique solutions iff the matrix is nonsingular.
\end{proof}

\subsection{Terminating a multiport by its adjoint through a gyrator}
A conventional linear circuit simulator can process only linear circuits with proper
device characteristics. Firther, unless the circuit has a unique solution,
the simulator will return an error message.
A useful artifice for processing a multiport through a conventional 
circuit simulator,
 is to terminate it appropriately so that the 
resulting network, if it has a solution, has a unique solution. This solution would also 
contain a solution to the original multiport.
We now describe this technique in detail.

Let $K,K"$ be representative matrices of vector spaces $\V,\V^*,$ respectively.
Consider the equation
\begin{align}
\label{eqn:uniquesol1}
\ppmatrix{K\\
K^{"}}\ppmatrix{
x}& =  \ppmatrix{s\\0}.
\end{align}

First note that the coefficient matrix, in Equation \ref{eqn:uniquesol1},
has number of rows equal to $r(\V)+r(\V^*)$ (Theorem \ref{thm:perpperp}) which is the 
number of columns of the matrix. Next, suppose the rows are linearly 
dependent. This would imply that a non trivial linear combination of
the rows is the zero vector, which, since the rows of $K,K"$ are linearly independent, in turn implies that a nonzero vector
$x,$ lies in the intersection of complementary orthogonal complex vector spaces, i.e., satisfies  $\langle x,x \rangle=0,$ a contradiction. 
%(In the complex case $x^Tx=0$ leads to $x=0$
%only if we take $x^*,$  the conjugate transpose of $x,$ in place of $x^T.$ Therefore the dot 
%product must be defined to be inner product for this argument to work.)
Thus, if a matrix is made up of  two sets of rows, 
which are representative matrices of complementary orthogonal
vector spaces,
then it must be nonsingular.
Therefore, the coefficient matrix in Equation \ref{eqn:uniquesol1}
is nonsingular.

%\begin{example}
Let the multiport behaviour $\breve{\A}_{PP"}$ be the solution space 
of the equation $Bv_{P}-Qi_{P"}=s,$ with linearly independent 
rows and let $\breve{\V}_{PP"}$ 
be the solution space of the equation $Bv_{P}-Qi_{P"}=0.$
 Let the dual multiport behaviour $\breve{\V}^{*}_{PP"}$ be the solution
space of equation $B^{"}v_{P}-Q^{"}i_{P"}=0,$
where the rows of $(B^{"}|-Q^{"})$
form a basis for the space complementary orthogonal to the 
row space of  $(B|-Q).$

The constraints of the two multiport behaviours together give the following
equation.
\begin{align}
\label{eqn:primaldual}
\ppmatrix{
        B  &\vdots &  -Q \\
        B^{"} & \vdots &       -Q^{"}}\ppmatrix{v_{P}\\i_{P"}}&=\ppmatrix{s\\ 0}. 
\end{align}
The first and second set of rows of the coefficient matrix of the above 
equation are  linearly independent and span complementary orthogonal spaces. Therefore the coefficient
matrix is invertible and the equation has a unique solution.
Now suppose we manage to terminate $\N_P$ by another multiport $\N^2_{{P}}$
in such a way that the port voltage and current vectors of $\N_P$ satisfy an equation of the kind  (\ref{eqn:primaldual}) above. Then the port voltage and current 
vectors of $\N_P$ would be unique and would uniquely determine the port voltage and current 
vectors of $\N^2_{{P}}.$ If  both multiports are rigid,
this would also uniquely determine the internal voltage and current vectors of
 both $\N_P$ and $\N^2_{{P}},$ which means that the network, obtained 
by terminating $\N_P$ by $\N^2_{{P}},$
 has a unique solution. We show below that we can build $\N^2_{{P}}$
by first building $\N^{adj}_{\tilde{P}}$ and attaching $1:1$ gyrators to its ports.

Let  $\mathcal{N}_P$ be a rigid multiport on graph $\G_{SP}$
and device characteristic ${\A}_{SS"},$ and let it have the port behaviour 
$\breve{\A}_{PP"}.$  
Let $\mathcal{N}^{adj}_{\tilde P}$ be on the copy $\G_{\tilde{S}\tilde{P}}$ of $\G_{SP},$ with device characteristic $({\V}^{adj}_{SS"})_{\tilde{S}\tilde{S}"}.$
By the definition of rigidity, 
the port behaviour $\breve{\A}_{PP"}$ is nonvoid and by Theorem  \ref{thm:regularrecursive}, $\mathcal{N}^{adj}_{\tilde P}$ is also
rigid.
We will now show that the network $[\mathcal{N}_P\oplus \mathcal{N}^{adj}_{\tilde P}]\cap \g^{{P}\tilde{P}}$
 has a unique solution.

 Let ${\V}_{SS"}$  be the vector space translate of 
${\A}_{SS"}$ and $\breve{\V}_{PP"}$  be that  of 
$\breve{\A}_{PP"}.$ 
We note that\\ 
$(\V^v(\G_{SP})\oplus (\V^i(\G_{SP}))_{S"P"})\lrar {\A}_{SS"}
%=(\V^v(\G_{SP})\oplus (\V^v(\G_{SP}))^{*}_{S"P"})\lrar {\A}_{SS"}
= (\breve{\A}_{PP"})_{P(-P")}.$
Thus the constraints of $\N_P$ are  
equivalent, as far as the variables $(v_P,i_{P"})$
are concerned, to the first set of equations of Equation \ref{eqn:primaldual}.
\\Next, 
$(\V^v(\G_{SP})\oplus (\V^i(\G_{SP}))_{S"P"})\lrar {\V}_{SS"}=
(\breve{\V}_{PP"})_{P(-P")},$ 
using Theorem \ref{thm:IIT2}. 
By Corollary \ref{cor:adjointmultiport},
 the port behaviour of  $\mathcal{N}^{adj}_{\tilde P}$ is $\breve{\V}^{adj}_{\tilde{P}\tilde{P}"}\equivd (\breve{\V}^{adj}_{PP"})_{\tilde{P}\tilde{P}"}.$ 
Further, 
$(\breve{\V}^{adj}_{\tilde{P}\tilde{P}"})_{{P}"(-{P})}=
(\breve{\V}^{adj}_{\tilde{P}\tilde{P}"})_{(-{P}"){P}}=
\breve{\V}^{*}_{PP"}.$
Thus the constraints of $\mathcal{N}^{adj}_{\tilde P}$ together 
with the constraint $ v_{P}=-i_{\tilde{P}"}; i_{P"}=v_{\tilde{P}}$
(the gyrator $\g^{{P}\tilde{P}}$),
are equivalent, as far as the variables $(v_P,i_{P"})$
 are concerned, to the defining equations of $\breve{\V}^{*}_{PP"},$
 i.e., to 
 the second set of equations of Equation \ref{eqn:primaldual}.

Thus, the constraints of $[\mathcal{N}_P\oplus \mathcal{N}^{adj}_{\tilde P}]\cap \g^{{P}\tilde{P}},$ are equivalent as far as the variables $(v_P,i_{P"})$
 are concerned, to
the constraints of Equation \ref{eqn:primaldual}.
 We have seen that this equation has a unique solution, say $(\hat{v}_P, -\hat{i}_{P"}).$
%Now $(\hat{v}_P, -\hat{i}_{P"})\in \hat{\A}_{PP"}.$ Therefore,
%$(\hat{v}_P, \hat{i}_{P"})$ is the restriction of a solution of $\N_P$
%to $P\uplus P".$ By the rigidity of $\N_P,$ there is a unique 
%solution $(\hat{v}_S,\hat{v}_P,\hat{i}_{S"},\hat{i}_{P"})$ of $\N_P.$
Using the gyrator constraints $ v_{P}=-i_{\tilde{P}"}; i_{P"}=v_{\tilde{P}},$
 we get a corresponding vector $(\tilde{v}_{\tilde{P}}, \tilde{i}_{\tilde{P}"})$
that is the restriction of a solution of  $\mathcal{N}^{adj}_{\tilde P}$
to $\tilde{P}\uplus \tilde{P}".$
Now $(\hat{v}_P, -\hat{i}_{P"})\in \hat{\A}_{PP"}.$ Therefore,
$(\hat{v}_P, \hat{i}_{P"})$ is the restriction of a solution of $\N_P$
to $P\uplus P".$ By the rigidity of $\N_P,$ there is a unique
solution $(\hat{v}_S,\hat{v}_P,\hat{i}_{S"},\hat{i}_{P"})$ of $\N_P.$
By the rigidity of $\mathcal{N}^{adj}_{\tilde P},$ there is a unique
solution $(\tilde{v}_{\tilde{S}},\tilde{v}_{\tilde{P}},\tilde{i}_{\tilde{S}"},\tilde{i}_{\tilde{P}"})$ of $\mathcal{N}^{adj}_{\tilde P}.$
Thus the vector $(\hat{v}_S,\hat{v}_P,\tilde{v}_{\tilde{S}},\tilde{v}_{\tilde{P}},\hat{i}_{S"},\hat{i}_{P"}, \tilde{i}_{\tilde{S}"},\tilde{i}_{\tilde{P}"})$
is the unique solution of $[\mathcal{N}_P\oplus \mathcal{N}^{adj}_{\tilde P}]\cap \g^{{P}\tilde{P}}.$

The device characteristic $\A_{SS"}\oplus {\V}^{adj}_{\tilde{S}\tilde{S}"}\oplus \g^{{P}\tilde{P}}$ is proper 
because $r(\V_{SS"}\oplus {\V}^{adj}_{\tilde{S}\tilde{S}"})= 
r(\V_{SS"})+r(\V^{*}_{SS"})=2|S|= |S|+|\tilde{S}|,$
 $r(\g^{{P}\tilde{P}})=2|P|=|P\uplus \tilde{P}|,$
so that dimension of $\V_{SS"}\oplus {\V}^{adj}_{\tilde{S}\tilde{S}"}\oplus \g^{{P}\tilde{P}}$  equals $|S|+|\tilde{S}|+|P\uplus \tilde{P}|.$
Since the network $[\mathcal{N}_P\oplus \mathcal{N}^{adj}_{\tilde P}]\cap \g^{{P}\tilde{P}}$
has proper device characteristic and also a unique solution,
our conventional circuit simulator can process it and obtain its solution.

%If the multiport $\mathcal{N}_P$ is not rigid,
%it may not have a solution and then the port behaviour $\breve{\A}_{PP"}$ would be  void.
%Even if the multiport has a solution, so that the behaviour $\breve{\A}_{PP"}$ is nonvoid,
%the above general procedure of solving 
%$[\mathcal{N}_P\oplus \mathcal{N}^{adj}_{\tilde P}]\cap \g^{{P}\tilde{P}},$
%will yield non unique internal voltages and currents in the multiports 
%$\mathcal{N}_P,\mathcal{N}^{adj}_{\tilde P}.$
%Therefore, $[\mathcal{N}_P\oplus \mathcal{N}^{adj}_{\tilde P}]\cap \g^{{P}\tilde{P}}$ will have  non unique solution.
%(In both the above cases our standard circuit simulator would give error
%messages.)
%Thus $[\mathcal{N}_P\oplus \mathcal{N}^{adj}_{\tilde P}]\cap \g^{{P}\tilde{P}}$ has a  unique solution iff $\N_P$  is rigid.

We have computed a single vector $x_{PP"}^{p}\equivd (v_{P},-i_{P"})\in \breve{\A}_{PP"}.$
We next consider the problem of finding a generating set for the vector space translate
$\breve{\V}_{PP"}$ of $\breve{\A}_{PP"}.$ \\Let $ \g_{tv}^{P\tilde{P}}$ denote the 
affine space that is the solution set of the constraints\\
$v_{e_j} = -i_{\tilde{e}_j"}, e_j\in P,j\ne t,  v_{e_t}+1 = -i_{\tilde{e}_t"}; i_{e_i}=v_{\tilde{e}_i"}, e_i\in P.$
\\Let $ \g_{ti}^{P\tilde{P}}$ denote the
affine space that is the solution set of the constraints\\
$v_{e_i} = -i_{\tilde{e}_i"}, e_i\in P; i_{e_j}=v_{\tilde{e}_j"}, e_j\in P,j\ne t, i_{e_t}+1=v_{\tilde{e}_t"}.
 $

Now solve $[\mathcal{N}^{hom}_P\oplus \mathcal{N}^{adj}_{\tilde P}]\cap \g_{tv}^{{P}\tilde{P}}$ for each $e_t\in P$ and $[\mathcal{N}^{hom}_P\oplus \mathcal{N}^{adj}_{\tilde P}]\cap \g_{ti}^{{P}\tilde{P}}$ for each $e_t"\in P"$
(see Figure \ref{fig:temp4}(b) and \ref{fig:temp4}(c)).
(We remind the reader that $\mathcal{N}^{hom}_P$ is obtained
from $\N_{P}$ by replacing its device characteristic $\A_{SS"}$ by
the vector space translate 
$\V_{SS"}.$)

We prove below, in Lemma \ref{lem:behaviourbasis},
 that each solution yields a vector in $\breve{\V}_{PP"}$
and the vectors corresponding to all $e_t\in P, e_t"\in P",$ form a generating set for $\breve{\V}_{PP"}.$

We summarize these steps in the following Algorithm.
%\begin{figure}
%\begin{center}
% \includegraphics[width=5.5in]{temp4.pdf}
% \caption{Computation of port behaviour of a multiport}
%
%\label{fig:temp4}
%%\caption{}
%% \caption{Example $RLCEJ$ Network and its multiport decomposition
%%into capacitive multiport and $RL$ multiport}
%\end{center}
%\end{figure}
%
%
\begin{algorithm}
\label{alg:TN}
Input: A multiport $ \N_P$ on $\G_{SP}$ with affine
device characteristic ${\A}_{SS"}.$\\
Output: The port behaviour $\breve{\A}_{PP"}$ of $ \N_P$ if
$ \N_P$ is rigid.\\
Otherwise a statement that $ \N_P$ is not rigid.

Step 1. Build the network $\N^{large}\equivd [\mathcal{N}_P\oplus \mathcal{N}^{adj}_{\tilde P}]\cap \g^{{P}\tilde{P}}$ on graph 
  $\G_{SP}\oplus \G_{\tilde{S}\tilde{P}}$ \\with device characteristic ${\A}_{SS"}\oplus {\V}^{adj}_{\tilde{S}\tilde{S}"}\oplus \g^{{P}\tilde{P}},$ 
where ${\V}_{SS"}$ is the vector space translate of ${\A}_{SS"}$
and ${\V}^{adj}_{\tilde{S}\tilde{S}"}\equivd ({\V}_{SS"}^{*})_{(-\tilde{S}")\tilde{S}}.$
(see Figure \ref{fig:temp4}(a).)
\\
Find the unique solution (if it exists) of $\N^{large}$ and restrict it to $P\uplus P"$
to obtain\\ $(v^p_{P},i^p_{P"}).$ The vector $(v^p_{P},-i^p_{P"})$ belongs to $ \breve{\A}_{PP"}.$\\
If no solution exists or if there are non unique solutions output `$\N_P$ not rigid' and\\ STOP.
\\
Step 2. Let $\mathcal{N}^{hom}_P$ be  obtained by replacing 
the device characteristic ${\A}_{SS"}$ by ${\V}_{SS"}$ in $\mathcal{N}_P.$
\\
For $t= 1, \cdots , |P|,$ build and solve $[\mathcal{N}^{hom}_P\oplus \mathcal{N}^{adj}_{\tilde P}]\cap \g_{tv}^{{P}\tilde{P}}$ and restrict it to $P\uplus P"$
to obtain $(v^{tv}_{P},i^{tv}_{P"}).$ (see Figure \ref{fig:temp4}(b).)

The vector $(v^{tv}_{P},-i^{tv}_{P"})\in \breve{\V}_{PP"}.$
\\
For $t= 1, \cdots , |P|,$ build and solve $[\mathcal{N}^{hom}_P\oplus \mathcal{N}^{adj}_{\tilde P}]\cap \g_{ti}^{{P}\tilde{P}}$ and restrict it to $P\uplus P"$
to obtain $(v^{ti}_{P},i^{ti}_{P"}).$ The vector $(v^{ti}_{P},-i^{ti}_{P"})\in \breve{\V}_{PP"}.$
(see Figure \ref{fig:temp4}(c).)
\\
Step 3. Let $\breve{\V}_{PP"}$ be the span of the vectors
$(v^{tv}_{P},-i^{tv}_{P"}), (v^{ti}_{P},-i^{ti}_{P"}), t= 1, \cdots ,|P|.$
\\
Output $\breve{\A}_{PP"}\equivd (v^p_{P},-i^p_{P"})+\breve{\V}_{PP"}.$\\
STOP
\end{algorithm}
We complete the justification of Algorithm \ref{alg:TN} in the following 
lemma. 
\begin{lemma}
\label{lem:behaviourbasis}
Let $\N_{P}\equivd (\Gsp, \A_{SS"})$  be rigid. Let $\breve{\A}_{PP"}$ be the port behaviour of $\N_{P}$ and let $\breve{\V}_{PP"}$ be the vector space translate of $\breve{\A}_{PP"}.$
Then the following hold.
\begin{enumerate}
\item  The network $\N^{large}\equivd [\mathcal{N}_P\oplus \mathcal{N}^{adj}_{\tilde P}]\cap \g^{{P}\tilde{P}}$ has a proper device characteristic and has a unique solution.
\item Each of the  networks $[\mathcal{N}^{hom}_P\oplus \mathcal{N}^{adj}_{\tilde P}]\cap \g_{tv}^{{P}\tilde{P}}, e_t\in P, [\mathcal{N}^{hom}_P\oplus \mathcal{N}^{adj}_{\tilde P}]\cap \g_{ti}^{{P}\tilde{P}}, e"_t\in P",$
has a unique solution and restriction 
of the solution to $P\uplus P"$ gives a vector $(v^{tv}_{P}, i^{tv}_{P"})$ such that
$(v^{tv}_{P}, -i^{tv}_{P"})\in \breve{\V}_{PP"}$
or a vector $(v^{ti}_{P}, i^{ti}_{P"})$ such that
$(v^{ti}_{P}, -i^{ti}_{P"})\in \breve{\V}_{PP"}.$

\item The vectors $(v^{tv}_{P}, -i^{tv}_{P"}),\  t=1, \cdots , |P|,
(v^{ti}_{P}, -i^{ti}_{P"}),\  t=1, \cdots , |P|,$ 
form a generating set for $\breve{\V}_{PP"}.$
\end{enumerate}
\end{lemma}
\begin{proof}
We only prove parts 2 and 3 since part 1 has already been shown.\\
2. If we replace the device characteristic 
%$\hat{\K}_{S'S"}\oplus \hat{\V}^{adj}_{S'S"} \oplus \g^{{P}\tilde{P}}$
of $ [\mathcal{N}_P\oplus \mathcal{N}^{adj}_{\tilde P}]\cap \g^{{P}\tilde{P}},$
or that of  $[\mathcal{N}^{hom}_P\oplus \mathcal{N}^{adj}_{\tilde P}]\cap \g_{tv}^{{P}\tilde{P}}$ or that of  $[\mathcal{N}^{hom}_P\oplus \mathcal{N}^{adj}_{\tilde P}]\cap \g_{ti}^{{P}\tilde{P}}$ by its vector space translate, i.e., by
${\V}_{SS"}\oplus {\V}^{adj}_{SS"} \oplus \g^{{P}\tilde{P}},$ 
we get the device characteristic of $[\mathcal{N}^{hom}_P\oplus \mathcal{N}^{adj}_{\tilde P}]\cap \g^{{P}\tilde{P}}.$ All three networks have the same graph 
$\G_{SP}\oplus \G_{\tilde{S}\tilde{P}}.$
% as are $ [\mathcal{N}_P\oplus \mathcal{N}^{adj}_{\tilde P}]\cap \g^{{P}\tilde{P}}$ as well as $[\mathcal{N}^{hom}_P\oplus \mathcal{N}^{adj}_{\tilde P}]\cap \g_{tv}^{{P}\tilde{P}}$ and $[\mathcal{N}^{hom}_P\oplus \mathcal{N}^{adj}_{\tilde P}]\cap \g_{ti}^{{P}\tilde{P}}.$\\ 
\\By part 5 of Theorem \ref{thm:regularrecursive},
 a linear network $\N$ with a proper device characteristic, has a unique solution iff $\N^{hom}$ has a unique solution.
 We know that $ [\mathcal{N}_P\oplus \mathcal{N}^{adj}_{\tilde P}]\cap \g^{{P}\tilde{P}}$ has a proper device characteristic and has a unique solution. Thus $[\mathcal{N}^{hom}_P\oplus \mathcal{N}^{adj}_{\tilde P}]\cap \g^{{P}\tilde{P}},$ 
has a unique solution and therefore also $[\mathcal{N}^{hom}_P\oplus \mathcal{N}^{adj}_{\tilde P}]\cap \g_{tv}^{{P}\tilde{P}}$ and $[\mathcal{N}^{hom}_P\oplus \mathcal{N}^{adj}_{\tilde P}]\cap \g_{ti}^{{P}\tilde{P}}.$ \\
The restriction of a solution $(\hat{v}_S,\hat{v}_P,\tilde{v}_{\tilde{S}},\tilde{v}_{\tilde{P}},\hat{i}_{S"},\hat{i}_{P"}, \tilde{i}_{\tilde{S}"},\tilde{i}_{\tilde{P}"})$ of $[\mathcal{N}^{hom}_P\oplus \mathcal{N}^{adj}_{\tilde P}]\cap \g_{tv}^{{P}\tilde{P}}$ or of \\$[\mathcal{N}^{hom}_P\oplus \mathcal{N}^{adj}_{\tilde P}]\cap \g_{ti}^{{P}\tilde{P}}$ to $S\uplus P\uplus S"\uplus P"$ gives a solution of $\mathcal{N}^{hom}_P.$ Its restriction to $ P\uplus P"$ gives $(\hat{v}_P, -\hat{i}_{P"})
.$  By Theorem \ref{thm:translatemultiport}, $(\hat{v}_P, \hat{i}_{P"})\in \breve{\V}_{PP"}.$
%
%The restriction of a solution $(v_{S},v_{P}, v_{\tilde{S}}, v_{\tilde{P}},
%i_{S"},i_{P"}, i_{\tilde{S}"}, i_{\tilde{P}"}) $ of $[\mathcal{N}^{hom}_P\oplus \mathcal{N}^{adj}_{\tilde P}]\cap \g_{tv}^{{P}\tilde{P}}$ or $[\mathcal{N}^{hom}_P\oplus \mathcal{N}^{adj}_{\tilde P}]\cap \g_{ti}^{{P}\tilde{P}}$
%to $S\uplus P\uplus S"\uplus P"$ gives a solution of $\mathcal{N}^{hom}_P.$
%Its restriction to $ P\uplus P"$ gives $(v_{P},i_{P"}).$  By Theorem \ref{thm:translatemultiport}, $(v_{P},-i_{P"})\in \breve{\V}_{PP"}.$
\\
%3. Suppose $(v^j_{P'}, -i^j_{P"})= \Sigma_t\lambda _t(v^t_{P'}, -i^t_{P"}), t\ne j.$
%Observe that in each of the vectors $(v^t_{P'}, -i^t_{P"}),t\ne j,$ we have 
%$v^t_{e_j'} = i^t_{e_j"},$ whereas in $(v^j_{P'}, -i^j_{P"}),$ we have
%$v^j_{e_j'}+1 = i^j_{e_j"},$ a contradiction.
%Thus the vectors $(v^t_{P'}, -i^t_{P"}),\  t=1, \cdots , |P|$
%are linearly independent. 
3. Let $\breve{\V}_{PP"}$ be the solution space of 
$Bv_{P} -Qi_{P"} =0$ 
and let $(\breve{\V}^{adj}_{PP"})_{\tilde{P}\tilde{P}"}=(\breve{\V}^{*}_{PP"})_{(-\tilde{P}")\tilde{P}},$
be the solution space of 
$Q^{"}v_{\tilde{P}}+B^{"}i_{\tilde{P}"}=0,$
where the row spaces of $(B|-Q), (B^{"}|-Q^{"})$ are complementary
orthogonal.
A vector $(v_P,v_{\tilde{P}}, -i_{P"},-i_{\tilde{P}"})$ being the restriction of a solution of $ [\mathcal{N}_P\oplus \mathcal{N}^{adj}_{\tilde P}]\cap \g_{tv}^{{P}\tilde{P}}$ to $P\uplus P"\uplus \tilde{P}\uplus \tilde{P}"$ 
is equivalent to $(v_P,v_{\tilde{P}}, i_{P"},i_{\tilde{P}"})$  being a solution to the equation
\begin{align}
\label{eqn:primaldual1}
\ppmatrix{
        B  &  -Q & 0 & 0\\
        0 &  0  & Q^{"} & B^{"}\\
I &  0  & 0 & I\\
0&  I&-I& 0}
\ppmatrix{v_{P}\\i_{P"}\\v_{\tilde{P}}\\i_{\tilde{P}"}}&=\ppmatrix{0\\ 0\\-I^t\\0
} 
\end{align}
and a vector 
$(v_P,v_{\tilde{P}}, -i_{P"},-i_{\tilde{P}"})$ being the restriction of a solution of
 $ [\mathcal{N}_P\oplus \mathcal{N}^{adj}_{\tilde P}]\cap \g_{ti}^{{P}\tilde{P}}$ to $P\uplus P"\uplus \tilde{P}\uplus \tilde{P}"$
is equivalent to $(v_P,v_{\tilde{P}}, i_{P"},i_{\tilde{P}"})$ being a solution to the equation
\begin{align}
\label{eqn:primaldual2}
\ppmatrix{
        B  &  -Q & 0 & 0\\
        0 &  0  & Q^{"} & B^{"}\\
I &  0  & 0 & I\\
0&  I&-I& 0}
\ppmatrix{v_{P}\\i_{P"}\\v_{\tilde{P}}\\i_{\tilde{P}"}}&=\ppmatrix{0\\ 0\\0\\-I^t
}, 
\end{align}
%where the row spaces of $(B|-Q), (B^{\*}|-Q^{\*})$ are complementary
%orthogonal and i
where $I^t$ denotes the $t^{th}$ column of a  $|P|\times |P|$ identity matrix.
In the variables $v_{P},i_{P"},$ Equation \ref{eqn:primaldual1} reduces to

\begin{align}
\label{eqn:primaldual3}
\ppmatrix{
        B  &  -Q \\
        B^{"} & -Q^{"}}
\ppmatrix{v_{P}\\i_{P"}}&=\ppmatrix{0\\-B^{"}I^t
},
\end{align}
and Equation \ref{eqn:primaldual2} reduces to

\begin{align}
\label{eqn:primaldual4}
\ppmatrix{
        B  &  -Q \\
        B^{"} & -Q^{"}}
\ppmatrix{v_{P}\\i_{P"}}&=\ppmatrix{0\\Q^{"}I^t
}.
\end{align}
It is clear that a vector belongs to $\breve{\V}_{PP"}$ iff it is a 
solution of 
\begin{align}
\label{eqn:primaldual5}
\ppmatrix{
        B  &  -Q \\
        B^{"} & -Q^{"}}
\ppmatrix{v_{P}\\i_{P"}}&=\ppmatrix{0\\x
},
\end{align}
for some  vector $x.$
The space of all such $x$ vectors is the column space of the matrix 
$(B^{"} |-Q^{"}).$ Noting that, for any matrix $K,$ the product  $KI^t$ is the $t^{th}$ column 
of $K,$ we see that the solutions, for $t=1, \cdots , |P|,$   
of Equations \ref{eqn:primaldual3}
and \ref{eqn:primaldual4}, span $\breve{\V}_{PP"}.$

\end{proof}
\begin{remark} 
If the multiport $\mathcal{N}_P$ is not rigid,
it may not have a solution and then the port behaviour $\breve{\A}_{PP"}$ would be  void.
Even if the multiport has a solution, so that the behaviour $\breve{\A}_{PP"}$ is nonvoid,
the above general procedure of solving
$[\mathcal{N}_P\oplus \mathcal{N}^{adj}_{\tilde P}]\cap \g^{{P}\tilde{P}},$
will yield non unique internal voltages and currents in the multiports
$\mathcal{N}_P,\mathcal{N}^{adj}_{\tilde P}.$
Therefore, $[\mathcal{N}_P\oplus \mathcal{N}^{adj}_{\tilde P}]\cap \g^{{P}\tilde{P}}$ will have  non unique solution.
(In both the above cases our conventional circuit simulator would give error
messages.)
Thus $[\mathcal{N}_P\oplus \mathcal{N}^{adj}_{\tilde P}]\cap \g^{{P}\tilde{P}}$ has a  unique solution iff $\N_P$  is rigid.
\end{remark} 
\section{Maximum power transfer for linear multiports}
\label{sec:maxpower}
The original version of the maximum power transfer theorem, states that 
a linear $1$-port transfers maximum power to a  load if 
the latter has value equal to the adjoint (conjugate transpose) of the Thevenin impedance of the $1$-port.
In the multiport case the Thevenin equivalent is an impedance matrix 
whose adjoint  has to be connected to the multiport for maximum power transfer. These should be regarded as restricted forms of the theorem since 
they do not handle the case where the Thevenin equivalent does not exist.

It was recognized early that a convenient way of studying maximum power
transfer is to study the port conditions for which such a transfer occurs
\cite{desoer1,narayananmp}. We will use this technique 
to obtain such port conditions for an  affine multiport behaviour.
In the general, not necessarily strictly passive, case,  we can only try to obtain
stationarity of power transfer, rather than maximum power transfer.
After obtaining these conditions we show that they are in fact
 achieved, if at all, when the 
multiport is terminated by its adjoint (which is easy to build), through an ideal transformer.
This means that the multiport behaviour 
need only be available as the port behaviour of a multiport $\N_P,$
and not explicitly, as an affine space $\breve{\A}_{PP"}.$ 

\subsection{Stationarity of power transfer for linear multiports}
\label{subsec:maxpower}
%We fix some notation to consider the complex field case that is needed to handle steady state sinusoidal analysis.
We fix some preliminary notation needed for the discussion of the maximum power transfer theorem.
For any matrix $M,$ we take  $\overline{M}$ to be the conjugate and
$M^*$ to be the conjugate transpose.
Our convention for the sign of power associated with a multiport
behaviour is
that when $(\breve{v}_{P},\breve{i}_{P"})\in \breve{\A}_{PP"},$
the  {\bf power absorbed} by the multiport behaviour  $\breve{\A}_{PP"}$
is $\langle \breve{v}_{P}, \breve{i}_{P"}\rangle + \langle \breve{i}_{P"}, \breve{v}_{P}\rangle = \overline{\breve{v}_{P}^T}\breve{i}_{P"}+\overline{\breve{i}_{P"}^T}\breve{v}_{P}.$ (We omit the scale factor $\frac{1}{2}$ for better readability of the expressions involved.) The power delivered by it, is therefore the negative of this quantity.

The maximum power transfer problem is
\begin{align}
\label{eqn:optprobcomplex}
\mbox{minimize}\ \  
%\ \ \breve{v}_{P}^*\breve{i}_{P"}+\breve{i}_{P}^*\breve{v}_{P"}(\equivd 
(\overline{\breve{v}_{P}^T}\breve{i}_{P"}+\overline{\breve{i}_{P"}^T}\breve{v}_{P})\\
%\mbox{under}\\
%\label{eqn:feasible}
(\breve{v}_{P},\breve{i}_{P"})\in \breve{\A}_{PP"},\ \mbox{or equivalently,}\ B\breve{v}_{P}-Q\breve{i}_{P"}=s.
\end{align}

%==========================================================

If $(\breve{v}^{stat}_{P},\breve{i}^{stat}_{P"}),$
is a stationary point for the optimization problem \ref{eqn:optprobcomplex}
, we
have
\begin{align}
\label{eqn:optprob31}
 (\breve{i}^{stat}_{P"})^T\overline{\delta\breve{v}}_{P}+ (\breve{v}^{stat}_{P})^T\overline{\delta\breve{i}}_{P"}
+ (\overline{\breve{i}^{stat}_{P"})}^T{\delta\breve{v}}_{P}+ (\overline{\breve{v}^{stat}_{P})}^T{\delta\breve{i}}_{P"}
&=0,&
\end{align}
for every  vector $(\delta\breve{v}_{P},\delta \breve{i}_{P"}),$
such that $({B}({\delta\breve{v}}_{P})-{{Q}}({\delta \breve{i}}_{P"}))=0.$
Therefore, Equation \ref{eqn:optprob31} is satisfied iff,  for some vector $\lambda,$
%\begin{align}
%\label{eqn:optprob33}
 $${((\breve{i}^{stat}_{P"})^T|(\breve{v}^{stat}_{P})^T)-\lambda^T(\overline{B}|-\overline{{Q}})}
=0.$$
Thus the stationarity at $(\breve{v}^{stat}_{P},\breve{i}^{stat}_{P"}),$ is equivalent to
$${((\breve{v}^{stat}_{P})^T|(\breve{i}^{stat}_{P"})^T)-\lambda^T(-\overline{Q}|\overline{{B}})}
=0.$$
Since we must have
$$B\breve{v}^{stat}_{P}-Q\breve{i}^{stat}_{P"}=s,$$
the stationarity condition reduces to
\begin{align}
\label{eqn:optprob51}
\ppmatrix{B&-Q}\ppmatrix{-Q^*\\B^*}\lambda=s,
\end{align}
The vector space translate, $\breve{\V}_{PP"}$ of $\breve{\A}_{PP"}$ is the solution space
of the equation, $B\breve{v}_{P}-Q\breve{i}_{P"}=0,$
  and $ (\breve{\V}_{PP"}^{adj})_{P(-P")}\equivd (\breve{\V}^{*}_{PP"})_{P"P},$ is the row space of
$(-\overline{Q}|\overline{B}).$\\
Thus, the stationarity condition says that $(\breve{v}^{stat}_{P},\breve{i}^{stat}_{P"})$
belongs to $\breve{\A}_{PP"}\cap(\breve{\V}_{PP"}^{adj})_{P(-P")}.$\\
%To obtain the complex version of the statement of Theorem \ref{thm:maxpowerport} and its proof,
%$\V^{\perp}_{SS"}$ should be interpreted as $\V^*_{SS"},$
% $\V^{adj}_{SS"}$ should be taken to be
%$ (\V^*_{SS"})_{(-S")S},$
%where $\V^*_{SS"}$ is the collection of all vectors whose inner product
%with every vector in $\V_{SS"}$ is zero.
In the case where the multiport has a Thevenin impedance  $Z$ and Thevenin voltage $E,$ the above 
stationarity condition reduces to the condition $(Z+Z^*)(-\breve{i}^{stat}_{P"})^T) =E$ (see example in Section \ref{sec:intro}).
%\begin{remark}
%n the case where the multiport has a Thevenin equivalent,
%its port behaviour would have the form
% \breve{v}_{P}=Z\breve{i}_{P"}+E,$ i.e., the form
%I \breve{v}_{P}-Z\breve{i}_{P"}=E,$
%o that the stationarity condition reduces to 
%\ppmatrix{I&-Z}\ppmatrix{-Z^*\\I}\lambda=E,$ i.e., to 
%(Z+Z^*)(-\lambda) =E.$ \\
%he condition ${((\breve{i}^{stat}_{P"})^T|(\breve{v}^{stat}_{P})^T)-\lambda^T(\overline{B}|-\overline{{Q}})}=0$ reduces to 
%{((\breve{i}^{stat}_{P"})^T|(\breve{v}^{stat}_{P})^T)-\lambda^T
%I|-\overline{{Z}}})=0,$
%${((\breve{i}^{stat}_{P"})^T|(\breve{v}^{stat}_{P})^T)-\lambda^T(I|-\overline{Z})=0,$ 
%o that $-\lambda =-\breve{i}^{stat}_{P"}
% and $\breve{v}^{stat}_{P}= Z^* (-\breve{i}^{stat}_{P"}).$ Thus $-\lambda$ is the current flowing out of the multiport when we terminate 
%t by $Z^*.$
%end{remark}

We note that, even when the multiport is rigid, 
%the vector $\lambda$ has $|P|$
%entries and the coefficient matrix in Equation \ref{eqn:optprob5} is
 %$|P|\times |P|.$ But 
%
%===================================================================
%
%changes in other files
%
%===================================================================
%
%
Equation \ref{eqn:optprob51} may have no solution, in which case we have no stationary vectors
for power transfer. 
%(eg. when the coefficient matrix is square  but singular) 
If the equation has a unique solution,  using that $\lambda $ vector
we get a unique stationary vector $(\breve{v}^{stat}_{P}, \breve{i}^{stat}_{P"}).$
We next show that the stationarity condition is achieved at the ports of the multiport $\N_P,$
if we terminate it by $\N^{adj}_{\tilde{P}}$ through the ideal transformer $\T^{P\tilde{P}}$ resulting in the network $[\N_P\oplus\N^{adj}_{\tilde{P}}]\cap \T^{P\tilde{P}},$ where the ideal transformer $\T^{P\tilde{P}}$ satisfies the 
equations  $v_{P}=v_{\tilde{P}}; i_{P}=-i_{\tilde{P}}.$

\begin{theorem}
\label{thm:maxpowerport}
Let  $\mathcal{N}_P,$ on graph $\G_{SP}$
and device characteristic ${\A}_{SS"},$ have the port behaviour
$\breve{\A}_{PP"}.$ 
Let ${\V}_{SS"}, \breve{\V}_{PP"}$ be the vector space translates of ${\A}_{SS"}, \breve{\A}_{PP"},$
respectively.
Let $\mathcal{N}^{adj}_{\tilde P}$ be on the disjoint copy $\G_{\tilde{S}\tilde{P}}$ of $\G_{SP},$ with device characteristic ${\V}^{adj}_{\tilde{S}\tilde{S}"}\equivd ({\V}^{adj}_{SS"})_{\tilde{S}\tilde{S}"}.$ 
\begin{enumerate}
\item A vector $(v_{P},i_{P"})$ is the restriction of a solution of the network $\N^{large}\equivd [\mathcal{N}_P\oplus \mathcal{N}^{adj}_{\tilde P}]\cap \T^{{P}\tilde{P}}$ to $P\uplus P",$ iff 
$(v_{P},-i_{P"})\in \breve{\A}_{PP"}\cap(\breve{\V}_{PP"}^{adj})_{P(-P")}.$
\item  
%======================================================
%
%changes files
%
%==========================================================
%$\breve{v}_{P'}^T\breve{i}_{P"}.$ 
Let $(\breve{v}^{stat}_{P},\breve{i}^{stat}_{P"})\in
\breve{\A}_{PP"}.$ Then $(\breve{v}^{stat}_{P},\breve{i}^{stat}_{P"})$ satisfies 
stationarity condition with respect to the power absorbed by 
%$(\breve{v}_{P},\breve{i}_{P"}),$
$(\breve{v}_{P},\breve{i}_{P"})\in
\breve{\A}_{PP"}$ 
 iff $(\breve{v}^{stat}_{P},\breve{i}^{stat}_{P"})\in
\breve{\A}_{PP"}\cap(\breve{\V}_{PP"}^{adj})_{P(-P")}.$
\item Let $(v^{1}_{P},-i^{1}_{P"})$ be the restriction of a solution of the multiport $\N_P,$ to $P\uplus P".$
Then $(\breve{v}^{1}_{P},\breve{i}^{1}_{P"})$
satisfies the stationarity condition with respect to 
the power absorbed by 
%$(\breve{v}_{P},\breve{i}_{P"}),$
$(\breve{v}_{P},\breve{i}_{P"})\in
\breve{\A}_{PP"}$ 
%
%
%$\breve{v}_{P}^T\breve{i}_{P"},$
%$(\breve{v}_{P},\breve{i}_{P"})\in
%\breve{\A}_{PP"}$ 
iff $(v^{1}_{P},-i^{1}_{P"})$  is the restriction of a solution of the network
$[\mathcal{N}_P\oplus \mathcal{N}^{adj}_{\tilde P}]\cap \T^{{P}\tilde{P}},$
to $P\uplus P".$

\end{enumerate}
\end{theorem}
\begin{proof}
1. The restriction of the set of solutions of  $\mathcal{N}_P$ on graph $\G_{SP}$ to $P\uplus P",$
is\\ 
$[(\V_{SP}\oplus (\V^{*}_{SP})_{S"P"})\cap \A_{SS"}]\circ {PP"}.$
This is  the same as
$(\breve{\A}_{PP"})_{P(-P")}.$\\
The restriction of the set of solutions of $\mathcal{N}^{adj}_{\tilde P}$ on the disjoint copy $\G_{\tilde{S}\tilde{P}}$ of $\G_{SP},$ to $\tilde{P}\uplus \tilde{P}"$ is\\
$[(\V_{\tilde{S}\tilde{P}}\oplus (\V^{*}_{\tilde{S}\tilde{P}})_{\tilde{S}"\tilde{P}"})\cap {\V}^{adj}_{\tilde{S}\tilde{S}"}]\circ {\tilde{P}\tilde{P}"}.$ 
This we know, by Corollary \ref{cor:adjointmultiport},
 to be the same as
$(\breve{\V}^{adj}_{PP"})_{\tilde{P}-\tilde{P}"}.$\\
%Let $\mathcal{N}^{adj}_{P}$ be on  $\G_{SP},$ with device characteristic ${\V}^{adj}_{{S}{S}"}.$\\
The restriction of the set of solutions of $\mathcal{N}^{adj}_{P}$ to ${P}\uplus {P}"$ is 
$[(\V_{{S}{P}}\oplus (\V^{*}_{{S}{P}})_{{S}"{P}"})\cap {\V}^{adj}_{{S}{S}"}]\circ {{P}{P}"}.$\\
We thus have,
$[((\V_{\tilde{S}\tilde{P}}\oplus (\V^{*}_{\tilde{S}\tilde{P}})_{\tilde{S}"\tilde{P}"})\cap {\V}^{adj}_{\tilde{S}\tilde{S}"})\cap \T^{P\tilde{P}}]\circ {PP"}
=[((\V_{\tilde{S}\tilde{P}}\oplus (\V^{*}_{\tilde{S}\tilde{P}})_{\tilde{S}"\tilde{P}"})\cap {\V}^{adj}_{\tilde{S}\tilde{S}"})\circ {\tilde{P}\tilde{P}"}]_{P(-P")}$\\$= ([(\V_{{S}{P}}\oplus (\V^{*}_{{S}{P}})_{{S}"{P}"})\cap {\V}^{adj}_{{S}{S}"}]\circ {{P}{P}"})_{P(-P")}=\breve{\V}^{adj}_{PP"},$
since vectors in $ \T^{P\tilde{P}}$ are precisely the ones that satisfy $v_{P}=v_{\tilde{P}"}, i_{P"}=-i_{\tilde{P}"}
.$\\
%Clearly,
%$[((\V_{\tilde{S}\tilde{P}}\oplus (\V^{*}_{\tilde{S}\tilde{P}})_{\tilde{S}"\tilde{P}"})\cap {\V}^{adj}_{\tilde{S}\tilde{S}"})\circ {\tilde{P}\tilde{P}"}]_{PP"}=[(\V_{{S}{P}}\oplus (\V^{*}_{{S}{P}})_{{S}"{P}"})\cap {\V}^{adj}_{{S}{S}"}]\circ {{P}{P}"}=(\breve{\V}^{adj}_{PP"})_{P(-P")}.$
%Therefore, $[((\V_{\tilde{S}\tilde{P}}\oplus (\V^{*}_{\tilde{S}\tilde{P}})_{\tilde{S}"\tilde{P}"})\cap {\V}^{adj}_{\tilde{S}\tilde{S}"})\cap \T^{P\tilde{P}}]\circ {PP"}
%=[[(\V_{{S}{P}}\oplus (\V^{*}_{{S}{P}})_{{S}"{P}"})\cap {\V}^{adj}_{{S}{S}"}]\circ {{P}{P}"}]_{P'(-P")}$\\$= \breve{\V}^{adj}_{PP"},$
%since vectors in $ \T^{P\tilde{P}}$ are precisely the ones that satisfy $v_{P}=v_{\tilde{P}"}, i_{P"}=-i_{\tilde{P}"}
%.$\\
%========================================================\\
The restriction of the set of solutions of $[\mathcal{N}_P\oplus \mathcal{N}^{adj}_{\tilde P}]\cap \T^{{P}\tilde{P}},$
to $P\uplus P"$ is therefore equal to \\
$[[(\V_{SP}\oplus (\V^{*}_{SP})_{S"P"})\cap \A_{SS"}]\circ {PP"}]
\cap [[(\V_{\tilde{S}\tilde{P}}\oplus (\V^{*}_{\tilde{S}\tilde{P}})_{\tilde{S}"\tilde{P}"})\cap {\V}^{adj}_{\tilde{S}\tilde{S}"}\cap \T^{P\tilde{P}}]\circ {{P}{P}"}]$
\\
$=(\breve{\A}_{PP"})_{P(-P")}\cap\breve{\V}_{PP"}^{adj}.$
%Thus a vector $(v_{P},i_{P"})$ is the restriction of a solution of the network\\ $\N^{large}\equivd [\mathcal{N}_P\oplus \mathcal{N}^{adj}_{\tilde P}]\cap \T^{{P}\tilde{P}}$ to $P\uplus P",$ iff
%$(v_{P},-i_{P"})\in \breve{\A}_{PP"}\cap(\breve{\V}_{PP"}^{adj})_{P(-P")}.$
%==========================================================
%Let $(v_{S'},v_{P'},v_{\tilde{S}'},v_{\tilde{P}'}, i_{S"},i_{P"},i_{\tilde{S}"},i_{\tilde{P}"})$ be a solution of 
%%$[\mathcal{N}_P\oplus \mathcal{N}^{adj}_{\tilde P}]\cap \T^{{P}\tilde{P}},$ 
%i.e., belong to \\$[(\V_{S'P'}\oplus (\V^{\*}_{S'P'})_{S"P"})\cap \A_{S'S"}\cap (\V_{\tilde{S}'\tilde{P}'}\oplus (\V^{\perp}_{\tilde{S}'\tilde{P}'})_{\tilde{S}"\tilde{P}"})\cap {\V}^{adj}_{\tilde{S}'\tilde{S}"}]\cap \T^{P\tilde{P}}.$\\
%Its restriction to $P'\uplus P",$ i.e.,   $(v_{P'},i_{P"})$ therefore belongs
%to \\ 
%both $[(\V_{S'P'}\oplus (\V^{\perp}_{S'P'})_{S"P"})\cap \A_{S'S"}]\circ {P'P"},$ i.e., to $(\breve{\K}_{P'P"})_{P'(-P")}$ \\
%as well as to
%$(((\V_{\tilde{S}'\tilde{P}'}\oplus (\V^{\perp}_{\tilde{S}'\tilde{P}'})_{\tilde{S}"\tilde{P}"})\cap {\V}^{adj}_{\tilde{S}'\tilde{S}"})\cap \T^{P\tilde{P}})\circ {{P}'{P}"}
%= \breve{\V}^{adj}_{P'P"}.$\\
%Thus 
%$(v_{P'},i_{P"})\in (\breve{\K}_{P'P"})_{P'(-P")}\cap \breve{\V}^{adj}_{P'P"},$
%i.e., $(v_{P'},-i_{P"})\in \breve{\K}_{P'P"}\cap (\breve{\V}^{adj}_{P'P"})_{P'(-P")}.$
%\\

Part 2 follows from  the discussion preceding the theorem and part 3 follows from part 1 and part 2.
\end{proof}
\subsection{Maximum Power Transfer Theorem for passive multiports}
\label{subsec:mptpassive}
We show       below that the stationarity conditions of the previous subsection reduce       to maximum power transfer conditions when the multiport is passive.

The \nw{power absorbed} by a vector 
$(v_{E},i_{E"})$ is given by $\langle v_{E},i_{E"}\rangle +  \langle i_{E"},v_{E}\rangle.$ 
\\
A vector space $\V_{SS"}$ is \nw{passive}, iff 
the power absorbed by $(x_{S},y_{S"})$ is nonnegative, 
%$\langle x_{S},y_{S"} \rangle\geq 0,$ 
whenever $(x_{S},y_{S"})\in \V_{SS"}.$ 
It is \nw{strictly passive} iff 
the power absorbed by every nonzero vector in $\V_{SS"}$ is positive.
An affine space   $\A_{SS"}$  is (strictly) passive iff its vector space translate is (strictly) passive. 
A multiport is (strictly) passive iff its port behaviour  
is (strictly) passive.

We now have a routine result which links passivity of the device characteristic
of a multiport to its port behaviour.
(A \nw{cutset} of a graph is a minimal set of edges 
which when deleted increases the number of connected components of the graph, equivalently, in the case of a connected graph, is a minimal set of edges not contained in any cotree.)

\begin{lemma}
\label{lem:maxpower}
Let $\N_P$ be a multiport on graph $\G_{SP}$ with device characteristic 
$\A_{SS"}.$ 
\begin{enumerate}
\item If $\A_{SS"}$ is passive so is the port behaviour $\breve{\A}_{PP"}$ of $\N_P.$
\item If $\A_{SS"}$ is strictly passive and $P$ contains no loops or cutsets of $\G_{SP},$ then the port behaviour $\breve{\A}_{PP"}$ of $\N_P$
 is also strictly passive.
\end{enumerate}
\end{lemma}
\begin{proof}
1. We assume that $\breve{\A}_{PP"}$ is nonvoid. We have,\\
$\breve{\A}_{PP"}= ((\V_{SP}\oplus (\V^{*}_{SP})_{S"P"})\lrar \A_{SS"})_{P(-P")} ,$ where $\V_{SP}\equivd \V^v(\G_{SP}).$\\
By Theorem \ref{thm:tellegen}, $\V^{*}_{S"P"}= (\V^i(\G_{SP}))_{S"P"}.$
By Theorem \ref{thm:IIT2}, it follows that the  vector space translate 
of $\breve{\A}_{PP"}$ is 
$\breve{\V}_{PP"}= ((\V_{SP}\oplus (\V^{*}_{SP})_{S"P"})\lrar \V_{SS"})_{P(-P")} .$\\ 
%============================================================
%
%changes files
%
%============================================================
%
%where $\V_{S'P'}\equiv (\V^v(\G_{SP}))_{S'P'}.$\\
Let  $(v_{P},-i_{P"})$ belong to $ \breve{\V}_{PP"}.$
Then there exist $(v_{S},v_{P})\in \V_{SP}$ and $(i_{S"},i_{P"})\in (\V^{*}_{SP})_{S"P"},$
such that $(v_{S},i_{S"})\in \V_{SS"}.$\\
By the orthogonality of $\V_{SP},\V^{*}_{SP},$ it follows that
$\langle (v_{S},v_{P}),(i_{S"},i_{P"})\rangle =\langle v_{S},i_{S"}\rangle+\langle v_{P},i_{P"}\rangle=0,$ and \\
by the passivity of $ \V_{SS"},$ it follows that $\langle v_{S},i_{S"}\rangle+ \langle i_{S"},v_{S}\rangle  \geq 0.$
Therefore  $\langle v_{P},-i_{P"}\rangle + \langle -i_{P"},v_{P}\rangle\geq 0.$\\
2. Without loss of generality, we assume that the graph $\G_{SP}$ is connected.
If $P$  contains no cutset or circuit of $\G_{SP}, $ then $S$ contains both 
a tree as well as a cotree of $\G_{SP}.$
If the voltages assigned to the branches of a tree are zero, the branches 
in its complement will have zero voltage. Therefore $\V^v(\G_{SP})\times P=\V_{SP}\times P$
must necessarily be a zero vector space.
If the currents in the branches of a cotree are zero the branches
in its complement will have zero current. Therefore $(\V^i(\G_{SP})\times P)_{P"}=\V^{*}_{S"P"}\times P"$
must necessarily be a zero vector space.\\
Now let $(v_{P},-i_{P"})\in \breve{\V}_{PP"}, (v_{P},-i_{P"})\ne 0_{PP"}.$
As in part 1 above, there exist $(v_{S},v_{P})\in \V_{SP}$ and $(i_{S"},i_{P"})\in \V^{*}_{SP},$
such that $(v_{S},i_{S"})\in \V_{SS"}.$\\
Since $\V_{SP}\times P$
and $\V^{*}_{SP}\times P$ are both zero vector spaces,
we must have that $(v_{S},i_{S"})\ne 0_{SS"}$  so that,
by strict passivity of $\V_{SS"},$ we have  
$\langle v_{S},i_{S"}\rangle + \langle i_{S"},v_{S}\rangle > 0.$
Further $\langle (v_{S},v_{P}),(i_{S"},i_{P"})\rangle =0,$ so that
we can conclude $\langle v_{P},-i_{P"}\rangle + \langle -i_{P"},v_{P}\rangle > 0.$
\end{proof}

%We say a device characteristic $\A_{S'S"}$ or a multiport behaviour $\breve{\A}_{P'P"}$ is \nw{(strictly) passive} iff its vector space part is (strictly) passive. 
When a port behaviour is passive or strictly passive, by taking into account second order terms, 
we can show that  the stationarity condition implies a  maximum power delivery condition.

Let $(\breve{v}^{stat}_{P},\breve{i}^{stat}_{P"})\in \breve{\A}_{PP"},$ satisfy the stationarity condition 
$((\breve{v}^{stat}_{P})^T,(\breve{i}^{stat}_{P"})^T) = {\lambda}^T(-\overline{Q}|\overline{B}),$ for some ${\lambda}.$
For any $(\breve{v}_{P},\breve{i}_{P"})\in \breve{\A}_{PP"},$ i.e., such that ${B}\breve{v}_{P}-{Q}\breve{i}_{P"}=s,$ we can write  $(\breve{v}_{P}, \breve{i}_{P"})=(\breve{v}^{stat}_{P}+\Delta\breve{v}_{P},\breve{i}^{stat}_{P"}+\Delta\breve{i}_{P"}).$
We then have ${B}\Delta\breve{v}_{P}-{Q}\Delta\breve{i}_{P"}=0$ and 
$[\overline{B}\overline{(\Delta\breve{v}_{P})}-\overline{Q}\overline{(\Delta\breve{i}_{P"}})]=0.$
%Now  for any $(\breve{v}_{P}, \breve{i}_{P"})\in \breve{\A}_{PP"},$ 
%i.e., ${B}_{P}\breve{v}_{P}-Q}_{P"}\breve{i}_{P"}=s,$
%we can write  $(\breve{v}_{P}, \breve{i}_{P"})=(\breve{v}^{stat}_{P}+\Delta\breve{v}_{P},\breve{i}^{stat}_{P"}+\Delta\breve{i}_{P"}).$
Therefore  
$\langle \breve{v}_{P}, \breve{i}_{P"}\rangle + \langle \breve{i}_{P"}, \breve{v}_{P}\rangle =  
\langle \breve{v}_{P}, \breve{i}_{P"}\rangle + \langle \breve{i}_{P"},      \breve{v}_{P}\rangle - \overline{{{\lambda}}^T}[{B}{(\Delta\breve{v}_{P})}-{Q}{(\Delta\breve{i}_{P"}})]
- {{\lambda}^T}[\overline{B}\overline{(\Delta\breve{v}_{P})}-\overline{Q}\overline{(\Delta\breve{i}_{P"}})]$
\\$= 
\langle\breve{v}^{stat}_{P},\breve{i}^{stat}_{P"}\rangle+
\langle\breve{i}^{stat}_{P"},\breve{v}^{stat}_{P}\rangle+
 (\Delta \breve{v}_{P})^T\overline{\breve{i}^{stat}_{P"}}+ (\breve{v}^{stat}_{P})^T\overline{(\Delta \breve{i}_{P"})}+
(\breve{i}^{stat}_{P"})^T\overline{(\Delta \breve{v}_{P})}+ 
(\Delta \breve{i}_{P"})^T\overline{\breve{v}^{stat}_{P}}
%(\breve{v}^{stat}_{P})^T\overline{\Delta \breve{i}_{P"}}
$\\$
+\langle \Delta \breve{v}_{P}, \Delta \breve{i}_{P"}\rangle
+\langle \Delta \breve{i}_{P"}, \Delta \breve{v}_{P}\rangle
-
{{\lambda}^T}[\overline{B}\overline{(\Delta\breve{v}_{P})}-\overline{Q}\overline{(\Delta\breve{i}_{P"}})]
- \overline{{\lambda}^T}[{B}{(\Delta\breve{v}_{P})}-{Q}{(\Delta\breve{i}_{P"}})].$
\\ We can rewrite the right side as  \\ 
$\langle\breve{v}^{stat}_{P},\breve{i}^{stat}_{P"}\rangle+ \langle \Delta \breve{v}_{P}, \Delta \breve{i}_{P"}\rangle +((\breve{v}^{stat}_{P})^T+{\lambda}^T\overline{Q})\overline{\Delta\breve{i}_{P"}}+((\breve{i}^{stat}_{P"})^T-{\lambda}^T\overline{B})\overline{\Delta\breve{v}_{P}}$\\
$+\langle\breve{i}^{stat}_{P"},\breve{v}^{stat}_{P}\rangle+ \langle \Delta \breve{i}_{P"}, \Delta \breve{v}_{P}\rangle +(\overline{(\breve{v}^{stat}_{P})^T}+\overline{{\lambda}^T}{Q}){\Delta\breve{i}_{P"}}+
(\overline{(\breve{i}^{stat}_{P"})^T}-\overline{{\lambda}^T}{B}){\Delta\breve{v}_{P}}.$\\
%((\breve{i}^{stat}_{P"})^T-{\lambda}^T\overline{B}_{P}){\Delta\breve{v}_{P}}. $\\
Applying the condition 
$((\breve{v}^{stat}_{P})^T,(\breve{i}^{stat}_{P"})^T) = {\lambda}^T(-\overline{Q}|\overline{B}),$ 
this expression reduces to\\
$\langle\breve{v}^{stat}_{P},\breve{i}^{stat}_{P"}\rangle+
\langle\breve{i}^{stat}_{P"},\breve{v}^{stat}_{P}\rangle
+\langle \Delta \breve{v}_{P}, \Delta \breve{i}_{P"}\rangle +
\langle \Delta \breve{i}_{P"}, \Delta \breve{v}_{P}\rangle 
.$
Therefore, 
$\langle \breve{v}_{P}, \breve{i}_{P"}\rangle + \langle \breve{i}_{P"}, \breve{v}_{P}\rangle$\\
$=\langle\breve{v}^{stat}_{P},\breve{i}^{stat}_{P"}\rangle+
\langle\breve{i}^{stat}_{P"},\breve{v}^{stat}_{P}\rangle
+\langle \Delta \breve{v}_{P}, \Delta \breve{i}_{P"}\rangle +
\langle \Delta \breve{i}_{P"}, \Delta \breve{v}_{P}\rangle 
.$

If $\breve{\A}_{PP"}$ is passive we have $\langle \Delta \breve{v}_{P}, \Delta \breve{i}_{P"}\rangle+\langle \Delta \breve{i}_{P"}, \Delta \breve{v}_{P}\rangle\geq 0 ,$
so that\\ 
$\langle \breve{v}_{P}, \breve{i}_{P"}\rangle + \langle \breve{i}_{P"}, \breve{v}_{P}\rangle\geq 
\langle\breve{v}^{stat}_{P},\breve{i}^{stat}_{P"}\rangle+ \langle\breve{i}^{stat}_{P"},\breve{v}^{stat}_{P}\rangle.$\\
Equivalently, the power delivered 
by $\breve{\A}_{PP"},$ is maximum when $(\breve{v}_{P}, \breve{i}_{P"})= (\breve{v}^{stat}_{P},\breve{i}^{stat}_{P"}).$
\\If $\breve{\A}_{PP"}$ is strictly passive we have
$\langle \breve{v}_{P}, \breve{i}_{P"}\rangle + \langle \breve{i}_{P"}, \breve{v}_{P}\rangle >
\langle\breve{v}^{stat}_{P},\breve{i}^{stat}_{P"}\rangle+ \langle\breve{i}^{stat}_{P"},\breve{v}^{stat}_{P}\rangle$\\
whenever $(\Delta \breve{v}_{P}, \Delta \breve{i}_{P"})\ne 0,$
so that 
$(\breve{v}^{stat}_{P},\breve{i}^{stat}_{P"})$
is the unique maximum delivery vector in $\breve{\A}_{PP"}.$

By Lemma \ref{lem:maxpower}, if a multiport $\N_P$ is passive,
so is its port behaviour. Therefore, Equation \ref{eqn:optprob51}
also gives the condition for maximum power transfer from a passive $\N_P.$

We thus have, from the above discussion, using Theorem \ref{thm:maxpowerport},
the following result.
%
%====================================================
%
%changes files
%
%=============================================
%
\begin{theorem}
\label{thm:passivemaxpowerport}
Let  $\mathcal{N}_P,$ on graph $\G_{SP}$
and with passive device characteristic ${\A}_{SS"},$ have the port behaviour
$\breve{\A}_{PP"}.$
Let ${\V}_{SS"}$ be the vector space translate of ${\A}_{SS"}.$
Let $\mathcal{N}^{adj}_{\tilde P}$ be on the disjoint copy $\G_{\tilde{S}\tilde{P}}$ of $\G_{SP},$ with device characteristic ${\V}^{adj}_{\tilde{S}\tilde{S}"}.$
\begin{enumerate}
%\item Let $(v_{P'},i_{P"})$ be the restriction of a solution of the network $\N^{large}\equiv [\mathcal{N}_P\oplus \mathcal{N}^{adj}_{\tilde P}]\cap \T^{{P}\tilde{P}}$ to $P'\uplus P".$
%Then $(v_{P'},-i_{P"})\in \breve{\A}_{P'P"}\cap(\breve{\V}_{P'P"}^{adj})_{P'-P"}.$
\item
%$\breve{v}_{P'}^T\breve{i}_{P"}.$ 
Let $(\breve{v}^{stat}_{P},\breve{i}^{stat}_{P"})\in
\breve{\A}_{PP"}.$ Then $(\breve{v}^{stat}_{P},\breve{i}^{stat}_{P"})$ satisfies
\begin{align}
\label{eqn:maximize1}
 maximize \ \ \langle \breve{v}_{P}, -\breve{i}_{P"}\rangle + \langle -\breve{i}_{P"}, \breve{v}_{P}\rangle,\ \ \ \ \ \
(\breve{v}_{P},\breve{i}_{P"})\in
\breve{\A}_{PP"},
\end{align}
 iff $(\breve{v}^{stat}_{P},\breve{i}^{stat}_{P"})\in
\breve{\A}_{PP"}\cap(\breve{\V}_{PP"}^{adj})_{P(-P")}.$
\item Let $(\breve{v}^{stat}_{P},-\breve{i}^{stat}_{P"})$ be the restriction of a solution of the multiport $\N_P,$ to $P\uplus P".$
Then $(\breve{v}^{stat}_{P},\breve{i}^{stat}_{P"})$
satisfies the optimization condition in Equation \ref{eqn:maximize1},
%the maximization condition with respect to $(-\breve{v}_{PT\breve{i}_{P"}),$
%$(\breve{v}_{P},\breve{i}_{P"})\in
%\breve{\A}_{P'P"}$ 
iff $(\breve{v}^{stat}_{P},-\breve{i}^{stat}_{P"})$  is the restriction of a solution of the network
$[\mathcal{N}_P\oplus \mathcal{N}^{adj}_{\tilde P}]\cap \T^{{P}\tilde{P}},$
to $P\uplus P".$
\end{enumerate}
\end{theorem}
\begin{remark}
%\begin{enumerate}
 Let  $\breve{\V}_{PP"}$ be strictly passive.
If $(0_{P},i_{P"})$ or $(v_{P},0_{P"})$ belongs to $\breve{\V}_{PP"},$
its strict passivity is violated. Therefore $r(\breve{\V}_{PP"}\times P)$
as well as $r(\breve{\V}_{PP"}\times P")$ must be zero and consequently
$r(\breve{\V}_{PP"}\circ P")$
as well as $r(\breve{\V}_{PP"}\circ P)$ must equal $|P|,$ using
Theorem \ref{thm:dotcrossidentity}.
Let  $(C_{P}|E_{P"})$ be the representative matrix of
 $\breve{\V}_{PP"}.$ It is clear that $C_{P},E_{P"}$
are nonsingular so that by invertible row transformation, we can reduce
$(C_{P}|E_{P"})$ to the form $(Z^T|I)$ or $(I|Y^T).$ We then have 
$\langle v_{P},i_{P"}\rangle + \langle i_{P"},v_{P}\rangle=
v_{P}^T\overline{i_{P"}}+i_{P"}^T\overline{v_{P}}
= i_{P"}^TZ^T\overline{i_{P"}}+i_{P"}^T\overline{Zi_{P"}}=
\langle i_{P"}, (Z+Z^*)i_{P"}\rangle.
$ The strict passivity is equivalent to $(Z+Z^*)$
being positive definite.
In this case $(Z+Z^*)$ is invertible and maximum transfer port condition is unique.
\end{remark}
\section{Conclusions}
\label{sec:conclusions}
%We have given a linear algebraic procedure, for 
%computing the port behaviour of general linear multiports, which is valid 
%even when the port behaviour is void.\\
We have given a method for computing the port behaviour 
of linear multiports using standard circuit
simulators which are freely available.
This will work for linear multiports  which have non void solutions for arbitrary internal source 
values and further have a unique internal solution corresponding to
a port condition consistent with the port behaviour.
The procedure involves  termination by the adjoint multiport through a gyrator.
The above procedure can be regarded as the most general form 
of Thevenin-Norton theorem  according to the view point 
%(a)  the theorem involves computation of port behaviour of general linear multiports or (b)
that the theorem involves the computation of the unique solutions of 
certain circuits which result by some termination of the linear multiport.

We have used  termination by the adjoint multiport through an ideal transformer
 to give a condition for 
maximum power transfer for general linear multiports.
This is the most general form of the Maximum power transfer theorem 
that is possible in terms of stationarity of power delivered.

{\bf Data availability statement}\\
The authors confirm that 
 there is no associated data for this manuscript.

{\bf Acknowledgements}\\
Hariharan Narayanan was partially supported by a Ramanujan Fellowship.

\appendix
%===============================================================
%
%changes in other files
%
%
%================================================================
%
%

%\end{remark}
%===================================================
%\bibliographystyle{elsarticle1-num}
%\bibliography{references}

\begin{thebibliography}{10}
\expandafter\ifx\csname url\endcsname\relax
  \def\url#1{\texttt{#1}}\fi
\expandafter\ifx\csname urlprefix\endcsname\relax\def\urlprefix{URL }\fi
\expandafter\ifx\csname href\endcsname\relax
  \def\href#1#2{#2} \def\path#1{#1}\fi

%\bibitem{belevitch68} V. ~Belevitch: {\it Classical Network 
%Theory}, Holden-Day, San Francisco, 1968. 

\bibitem{desoer1} C.~A.~Desoer, The Maximum Power Transfer Theorem 
for n-Ports, IEEE Transactions on Circuit Theory, (1973) 328--330.
%\bibitem{gantmacher1959theoryVol2}
%F.~Gantmacher, The Theory of Matrices: Vol.: 2, Chelsea publishing company,
%  1959.

\bibitem{desoerkuh} C.~A. ~Desoer \& E.~S.~Kuh, {\it Basic circuit theory}, McGraw-Hill, New York, 1969.
%\bibitem{edm65a} J. ~Edmonds: Minimum partition of a matroid
%into independent subsets, { Journal of Research of the National Bureau
%of Standards} {\bf 69B} (1965) 67-72.

\bibitem{forney2004}
G.~D. ~Forney \& M.~D. ~Trott, {The Dynamics of Group Codes : Dual Abelian Group
  Codes and Systems}, IEEE Transactions on Information Theory 50~(11) (2004)
  1--30.

%\bibitem{iritomizawa} M. ~Iri \&  N. ~Tomizawa, A practical criterion for the existence of the unique solution
%in a linear electrical network with mutual couplings, Trans. Inst. Electron. \& Commun.
%Eng. Jpn., 57-A, 8, (1974) 35--41.

%\bibitem{iritomizawa2} M. ~Iri \&  N. ~Tomizawa,
% A unifying approach to fundamental problems in network
%theory by means of matroids, Trans. Inst. Electron. \& Commun. Eng. Jpn., 58-A, No. I, (1975)
%33-40.

%\bibitem{irisurvey} M. ~Iri, Survey of recent trends in applications of matroids,
%Proc. IEEE ISCAS, Tokyo, (1979) 987.

%\bibitem{irireview} M. ~Iri,
% A review of recent work in Japan on principal partitions of matroids and their
%applications, Ann. New York Acad. Sci. 319 (1979) 306--319.

%\bibitem{irifujishige} M. ~Iri \&  S. ~Fujishige,
% Use of matroid theory in operations research, circuits and
%systems theory, Int. J. Syst. Sci., 12, (1981) 27-54.


%\bibitem{iriapplications} M. ~Iri,
%Applications of matroid theory, Mathematical Programming: The State 0f the Art,
%Springer, Berlin, (1983) 158--201.


%\bibitem{iriprogress}
%M. ~Iri,
%Structural Theory for the Combinatorial Systems Characterized by Submodular Functions,
%Chapter in Progress in Combinatorial Optimization
%(1984) 197--219, Academic Press.

%\bibitem{kalman}
%R.E.Kalman, {Mathematical description of linear dynamical systems }, SIAM J.
%  Control 1 (1963) 152--192.

\bibitem{kron39} G.~Kron, {\it Tensor Analysis of Networks},  J.Wiley,
New York, 1939.

\bibitem{kron63} G. Kron, {\it Diakoptics - Piecewise Solution
of Large Scale Systems}, McDonald,London,1963.

\bibitem{mayer} H.~F.~ Mayer, Ueber das Ersatzschema der Verstärkerröhre [On equivalent circuits for electronic amplifiers], Telegraphen- und Fernsprech-Technik (in German), 15: (1926) 335--337.
%\bibitem{murotairi1}  K.~Murota \&  M. ~Iri,
%Structural solvability of systems of equations —A mathematical formulation for distinguishing accurate and inaccurate numbers in structural analysis of systems,Japan Journal of Industrial and Applied Mathematics  2(1) (1985) 247--271.

%\bibitem{murotabook0} K. ~Murota, {\it Systems Analysis by Graphs and Matroids:
%Structural Solvability and Controllability}, Springer, Berlin, Heidelberg, New York, Tokyo, 1987.

\bibitem{murotabook} K. ~Murota, {\it Matrices and Matroids for Systems Analysis},
Springer, Berlin, Heidelberg, New York, Tokyo, 2000.


\bibitem{narayananmp}
H.~Narayanan, On the maximum power transfer theorem, 
Int. J. Elect. Enging Educ., 15 (1978) 161--167.

%\bibitem{HNarayanan1986a}
%H.~Narayanan, {On the decomposition of vector spaces}, Linear Algebra and its
%  Applications 76 (1986) 61--98.

\bibitem{HNarayananadjoint}
H.~Narayanan, {A Unified Construction of Adjoint Systems and Networks}, Circuit
  Theory and Applications 14 (1986) 263--276.

%\bibitem{narayanan1987topological}
%H.~Narayanan, {Topological transformations of electrical networks},
%  International Journal of Circuit Theory and Applications 15~(3) (1987)
%  211--233.



\bibitem{HNarayanan1997}
H.~Narayanan, {\it Submodular Functions and Electrical Networks}, Annals of Discrete
  Mathematics, vol. 54, North Holland, Amsterdam, 1997.

%\bibitem{narayanan2002some}
%H.~Narayanan, Some applications of an implicit duality theorem to connections
%  of structures of special types including dirac and reciprocal structures,
%  Systems \& Control Letters 45~(2) (2002) 87 -- 95.


\bibitem{HNarayanan2009}
H.~Narayanan, {\it Submodular Functions and Electrical Networks}, 
Open 2nd edition, 
 http://www.ee.iitb.ac.in/$\ \tilde{}$ hn/ , 2009. 

%\bibitem{HNPS2013}
%H.~Narayanan \&  H. ~Priyadarshan, 
%{A subspace approach to linear dynamical systems},
%Linear Algebra and its Applications 438 (2013) 3576--3599.

\bibitem{narayanan2016} H.~Narayanan, 
Implicit Linear Algebra and its Applications,
 arXiv:1609.07991v1 [math.GM] 
(2016).

\bibitem{norton} E.~L.Norton, Design of finite networks for uniform frequency characteristic, Bell Laboratories, Technical Report TM26--0--1860 (1926).

\bibitem{penfield}
P. Penfield Jr., R. Spence \&  S. Duinker, {\it Tellegen's Theorem and Electrical Networks}, Cambridge, Mass. M.I.T.Press, 1970.

%\bibitem{purslow} E.~J. Purslow, Solvability and analysis of linear active networks by use of 
% the state equations, IEEE Trans. Circuit Theory, CT-17 (1970) 469--475.
%\bibitem{recski89} A. Recski, {\it Matroid Theory and its Applications 
%in Electric Network Theory and in Statics}.
%(Springer-Verlag,Berlin,Heidelberg,New York,London,Paris,Tokyo,1989).

\bibitem{recski19} A. Recski, Hybrid Description and the Spectrum of Linear Multiports,  IEEE Transactions on Circuits and Systems II, 68 (2019) 1502--1506.

%\bibitem{STHN2014}
%Siva Theja \& H.~Narayanan,
%On the notion of generalized minor in topological network theory and matroids,
%Linear Algebra and its Applications 458 (2014) 1--46.



\bibitem{tellegen}
B.~D.~H. Tellegen, A general network theorem, with applications, Philips Res. Rept., 7 
(1952) 259--269.

\bibitem{thevenin0} L.~Thevenin, Sur un nouveau theoreme d'électricite dynamique [On a new theorem of dynamic electricity],
 CR des Séances de l'Académie des Sciences, 1883.
\bibitem{thevenin} L.~Thevenin, Extension of Ohm's law to complex electromotive circuits,
 Annales Telegraphiques, 1883.
%\bibitem{Trentelman2001}
%H.~L. Trentelman, A.~A. Stoorvogel \&  M.~L.~J. Hautus, Control Theory for Linear
%  Systems, Springer-Verlag, London, 2001.

%\bibitem{tutte} W.~T. Tutte, Lectures on matroids, { Journal of 
%Research of the National Bureau of Standards}, { B69} (1965) 1--48. 

\bibitem{schaft1999} A.~J.~van~der~Schaft, “Interconnection and geometry”, in The Mathematics of
Systems and Control: from Intelligent Control to Behavioral Systems, Editors
J.W. Polderman, H.L. Trentelman, Groningen, pp. 203–218, 1999.
%\bibitem{SchaftBook2000}
%A.{~van ~der Schaft}\&  J.~M. Schumacher, {\it An introduction to hybrid dynamical
%  systems}, Vol. 251 of Lecture Notes in Control and Information Sciences,
%  Springer, 2000.

%\bibitem{van2004bisimulation}A.~J.~van~der~Schaft,
% Bisimulation of dynamical systems, Hybrid Systems:
%  Computation and Control (2004) 291--294.

%
%\bibitem{schaft2004equivalence}
%A.~J.~van~der Schaft, {Equivalence of hybrid dynamical systems}, in: International
%  Symposium on Mathematical Theory of Networks and Systems, Vol.~16, Leuven:
%  Belgium, 2004.

%\bibitem{van2004equivalence}
%A.~J.~van ~der Schaft, Equivalence of dynamical systems by bisimulation, Automatic
%  Control, IEEE Transactions on 49~(12) (2004) 2160--2172.


\bibitem{vanderscaftport} A.~J.~van~der~Schaft, ~D.~Jeltsema, Port-Hamiltonian Systems Theory: An Introductory
Overview, Foundations and Trends 
in Systems and Control, vol. 1, no. 2-3,
(2014) 173--378.

%\bibitem{willemstrentelman} J.~C.~Willems,  H.~L.~Trentelman, On quadratic differential forms, SIAM Journal of Control and Optimization, Vol. 36, No. 5, (1998)  1703--1749, 1998.


%\bibitem{willems1991paradigms}
%J.~C.~Willems, {Paradigms and puzzles in the theory of dynamical systems}, IEEE
%  Trans. on Automatic Control 36~(3) (1991) 259--294.

%\bibitem{willems1997}
%J.~W. Polderman, J.~C. Willems, {\it Introduction to Mathematical Systems Theory: a
%  Behavioural Approach}, Springer-Verlag, 1997.

%\bibitem{Wonham1978}
%W.~M. Wonham, {\it Linear Systems Theory: A Geometric Approach}, Springer-Verlag, New
%  York, USA, 1978.


%\bibitem{basile1969controlled}
%G.~Basile, G.~Marro, {Controlled and conditioned invariant subspaces in linear
%  system theory}, Journal of Optimization Theory and Applications 3~(5) (1969)
%  306--315.

%\bibitem{wonham1970decoupling}
%W.~Wonham, A.~Morse, {Decoupling and pole assignment in linear multivariable
%  systems: a geometric approach}, SIAM J. Control 8~(1) (1970) 1--18.

%\bibitem{basile1969observability}
%G.~Basile, G.~Marro, {On the observability of linear, time-invariant systems
%  with unknown inputs}, Journal of optimization theory and applications 3~(6)
%  (1969) 410--415.

%\bibitem{basile1973new}
%G.~Basile, G.~Marro, {A new characterization of some structural properties of
%  linear systems: unknown-input observability, invertibility and functional
%  controllability†}, International Journal of Control 17~(5) (1973) 931--943.

%\bibitem{basile1982self}
%G.~Basile, G.~Marro, {Self-bounded controlled invariant subspaces: a
%  straightforward approach to constrained controllability}, Journal of
%  Optimization Theory and Applications 38~(1) (1982) 71--81.

%\bibitem{basile1986self}
%G.~Basile, G.~Marro, {Self-bounded controlled invariants versus
%  stabilizability}, Journal of Optimization Theory and Applications 48~(2)
%  (1986) 245--263.

%\bibitem{bhattacharyya1978observer}
%S.~Bhattacharyya, Observer design for linear systems with unknown inputs, IEEE
%  Trans. on Automatic Control 23~(3) (1978) 483 -- 484.

%\bibitem{fabian1975decoupling}
%E.~Fabian, W.~Wonham, Decoupling and disturbance rejection, IEEE Trans. on
%  Automatic Control 20~(3) (1975) 399 -- 401.

%S.~Bhattacharyya, On calculating maximal (a,b) invariant subspaces, IEEE Trans.
%  on Automatic Control 20~(2) (1975) 264 -- 265.

%\bibitem{thorp1973singular}
%J.~Thorp, {The singular pencil of a linear dynamical system†}, International
%  Journal of Control 18~(3) (1973) 577--596.




%\bibitem{HNarayananUnpub}
%H.~Narayanan, {On the Linking of Vector Spaces}, {Under Preparation}.

\end{thebibliography}

\end{document}